\documentclass{article}

\usepackage[light,sfmathbb,nowarning]{kpfonts}

\usepackage{authblk}
\usepackage{times}
\usepackage{soul}
\usepackage{url} 
\usepackage[hidelinks]{hyperref}
\usepackage[utf8]{inputenc}
\usepackage[small]{caption}
\usepackage{helvet}
\usepackage{courier}
\usepackage{amsfonts}
\usepackage{graphicx} %with the option draft figures are suppressed
\usepackage{mathdots}
\usepackage{natbib}
\usepackage{amsmath,amsthm,amsfonts}
\newtheorem{proposition}{Proposition}

\newtheorem{definition}{Definition}

\newtheorem{theorem}{Theorem}

\usepackage{color}
\usepackage{amsmath}
\usepackage{amssymb}
\usepackage{nicefrac}
\usepackage{amsthm}
\usepackage{booktabs}

\newcommand\mydots{\hbox to 0.8em{.\hss.\hss.}}

\newcommand{\reals}{{\mathbb{R}}}

\newcommand{\score}{\mathit{score}}
\newcommand{\av}{\mathit{av}}

\newcommand{\app}{{\mathrm{app}}}
\newcommand{\hamming}{\mathrm{ham}}
\DeclareMathOperator{\ham}{ham}

\usepackage{todonotes}

\newcommand{\kso}[1]{\todo[color=blue!30, inline]{KSo: #1}}

\newcommand{\np}{{{\mathrm{NP}}}}

\title{How to Sample Approval Elections?}
\date{}
\author[1,2]{Stanisław Szufa}
\author[1]{Piotr Faliszewski}
\author[1]{Łukasz Janeczko}
\author[3]{Martin Lackner}
\author[4]{Arkadii Slinko}
\author[1]{Krzysztof Sornat}
\author[5]{Nimrod Talmon}
\affil[1]{AGH University, Kraków, Poland}
\affil[2]{Jagiellonian University, Kraków, Poland}
\affil[3]{Technische Universität, Wien, Austria}
\affil[4]{University of Auckland, Auckland, New Zealand} 
\affil[5]{Ben Gurion University, Be’er Sheva, Israel}

\begin{document}

\maketitle

\pagestyle{plain}

\begin{abstract}
%   We extend the map-of-elections framework to the case of approval
%   elections. While doing so, we study a number of statistical
%   cultures, including some new ones, and we analyze their
%   properties. We find that approval elections can be understood in
%   terms of the average number of approvals in the votes, and the
%   extent to which the votes are chaotic.
    We study the multifaceted question of how to sample approval
    elections in a meaningful way. Our analysis aims to discern 
    the properties of various statistical cultures (both established 
    and new ones). Based on the map-of-elections framework 
    by~\citet{szu-fal-sko-sli-tal:c:map},
    we graphically represent statistical cultures; and, by that, provide an intuitive understanding of their differences and properties.
\end{abstract}

\section{Introduction}

In an approval election~\cite{bra-fis:b:approval-voting}, each voter
indicates which candidates he or she finds acceptable for a certain
task (e.g., to be a president, to join the parliament, or to enter the
final round of a competition) and a voting rule aggregates this
data into the final outcome.
%
% is used to aggregate these preferences and determine
% the winner(s).
%
In the single-winner setting (e.g., when choosing the president), the
most popular rule is to pick the candidate with the highest number of
approvals.  In the multiwinner setting (e.g., in parliamentary
elections or when choosing finalists in a competition) there is a rich
spectrum of rules to choose from, each with different properties and
advantages.
Approval voting is particularly attractive due to its simplicity and
low cognitive load imposed on the voters.
% even if there are many
% candidates to choose from, usually it is not too difficult to identify
% those few that one supports.
%
% (this stands in sharp
%contrast with ordinal voting where the voters has to rank all the
%candidates, including those that he or she has barely heard
%of).
%
Indeed,
its practical applicability %of approval voting
has already been tested in a number of field experiments, including
those in
France~\cite{las-str:j:approval-experiment,bau-ige:b:french-approval-voting,bou-bla-bau-dur-ige-lan-lar-las-leb-mer:t:french-approval-voting-2017}
and Germany~\cite{alo-gra:j:german-approval-voting}.  Over the recent
years there was also tremendous progress regarding its theoretical
properties (see, e.g., the overviews of
\citet{las-san:chapter:approval-multiwinner} and
\citet{lac:sko:t:approval-survey}).

% Practical applicability of approval voting has already been tested in
% a number of field experiments, including those in
% France~\cite{las-str:j:approval-experiment,bau-ige:b:french-approval-voting,bou-bla-bau-dur-ige-lan-lar-las-leb-mer:t:french-approval-voting-2017}
% and Germany~\cite{alo-gra:j:german-approval-voting}.  We are also
% witnessing remarkable progress regarding theoretical understanding of
% approval elections. For example, in recent years researchers
% introduced and analyzed various notions related to proportional
% representation~\cite{justifiedRepresentation,Sanchez-Fernandez2017Proportional,sko:c:prop-degree},
% studied strategyproofness~\cite{pet:c:approval-mw-strategy-proof} and
% monotonicity~\cite{san-fis:c:approval-monotonicity}, and provided
% strong axiomatic
% characterizations~\cite{lac-sko:j;multiwinner-approval-axioms}.  There
% is also an active line of work on the computational complexity of
% approval-based voting (see, e.g., the overview of
% \citet{bau-erd-hem-hem-rot:b:computational-apects-of-approval-voting}
% and the works of
% \citet{bar-gou-lan-mon-rie:c:possible-approval-winners},
% \citet{azi-gas-gud-mac-mat-wal:c:approval-multiwinner},
% \citet{sko-fal-lan:j:collective}, or
% \cite{bre-fal-kac-kno-nie:c:parameterized-collective}).

In spite of all these achievements, numerical experiments regarding
approval voting are still challenging to design. One of the main
difficulties is caused by the lack of consensus on which statistical
election models to use. Below we list a few models (i.e., statistical
cultures) that were recently used:

% there is no consensus 

% as to which statistical
% models of elections to use and how to set their parameters. As a
% consequence, many authors use different models and their results are
% not easy to compare:
\begin{enumerate}
\item In the impartial culture setting, we assume that each vote is
  equally likely. Taken literally, this means that each voter approves
  each candidate with
  probability~$\nicefrac{1}{2}$ \cite{bar-lan-yok:c:hamming-approval-manipulation}.
  As this is quite unrealistic, several authors treat the approval
  probability as a
  parameter~\cite{bre-fal-kac-nie2019:experimental_ejr,fal-sli-tal:c:vnw}
  or require that all voters approve the same (small) number of
  candidates~\cite{lac-sko:j:av-vs-cc}. A further refinement is to
  choose an individual approval probability for each
  candidate~\cite{lac-mal:c:vnw-shortlisting}.

\item In Euclidean models, each candidate and voter is a point in
  $\mathbb{R}^d$, where $d$ is a parameter, and a voter approves a
  candidate if they are sufficiently near. Such models are used, e.g.,
  by \citet{bre-fal-kac-nie2019:experimental_ejr} and
  \citet{god-bat-sko-fal:c:2d}. Naturally, the distribution of the
  candidate and voter points strongly affects the outcomes.
  % , or the
  % candidates and voters may have individual radii and a voter approves
  % a candidate if their distance is smaller than the sum of their
  % radii~\cite{god-bat-sko-fal:c:2d}.

\item Some authors consider statistical cultures designed for the
  ordinal setting (where the voters rank the candidates from the most
  to the least desirable one) and let the voters approve some
  top-ranked candidates (e.g., a fixed number of them). This approach
  is taken, e.g., by \citet{lac-sko:j:av-vs-cc} on top of the ordinal
  Mallows model (later on, \citet{caragiannis2022evaluating} provided
  approval-based analogues of the Mallows model).
\end{enumerate}
Furthermore, even if two papers use the same model, they often choose its
parameters differently. Since it is not clear how the parameters
affect the models, comparing the results from different papers is not
easy.
%
% Many authors also design ad-hoc models, meant to capture the specific
% phenomena that they study (see, e.g., the works of
% \citet{bar-lan-yok:c:hamming-approval-manipulation} and
% \citet{lac-mal:c:vnw-shortlisting}).  Worse yet, they often use
% different parameters for the models, so it is difficult to compare
% results from different papers.
%
% More importantly, it is not clear
% what parameters to use for all these models to capture realistic
% scenarios.
%
%
Our goal is to initiate a systematic study of approval-based
statistical cultures and attempt to rectify at least some of these
issues.  We do so by extending the map-of-elections framework of
\citet{szu-fal-sko-sli-tal:c:map} and
\citet{boe-bre-fal-nie-szu:c:compass} to the approval setting.

%\nimrod{maybe nicer to explain the map-of-elections more detailed next?}

Briefly put, a map-of-elections is a set of elections with a distance
between each pair.
% (since the elections are often generated randomly,
% the metric used to obtain these distances needs to be independent of
% voter and candidate renaming)\nimrod{I think the text in parenthesis
%   is only confusing}.
Such a set of elections can then be embedded in a
plane, by representing each election as a point and embedding the
distance between the elections into Euclidean distances in the
plane. Such a map-of-elections is useful as, by analyzing the
distances between elections generated from various models, we can
often obtain some insights about their nature.

%% The map framework relies on having a metric over elections, and it
%% requires this metric to be independent of renaming the candidates and
%% voters (the names are irrelevant in randomly generated elections).
%
% The first step %toward this goal
% is to choose a metric that
% % To this end, we need a metric that
% given two approval election (with the same number of candidates and
% the same number of voters) indicates how similar they are. Since our
% elections are sampled from statistical models, this metric must be
% independent of renaming the candidates and voters (these names do not
% carry any meaning in such elections).
% %
%
%%
%%We identify two such metrics, the {\em isomorphic Hamming distance} and the
%%{\em approvalwise distance}. 

To create a map-of-elections for approval elections,
we start by identifying two metrics between approval elections, 
the {\em isomorphic Hamming distance} and the
{\em approvalwise distance}. 
The first one is accurate, but difficult to compute,
whereas the second one is less precise, but easily computable.
Fortunately, %we find that 
%the elections that we consider 
in our election datasets
the two
metrics are strongly correlated; thus, we use the latter one.

Next, we analyze the space of approval elections with a given number
of candidates and voters.  For each $p \in [0,1]$, by $p$-identity
($p$-ID) elections we mean those
where all the votes are identical and approve the same $p$-fraction of
candidates;
and by $p$-impartial culture ($p$-IC) elections those where each voter chooses to approve each
candidate with probability $p$.  We view $p$-ID and
$p$-IC elections as two extremes on the spectrum of agreement between
the voters and, intuitively, we expect that every election (where each
voter approves on average a $p$ fraction of candidates) is located
somewhere between these two.  In particular, for $p, \phi \in [0,1]$,
we introduce the $(p,\phi)$-resampling model, which generates
elections whose expected approvalwise distance from $p$-ID is exactly
the $\phi$ fraction of the distance between $p$-ID and $p$-IC (and the
expected distance from $p$-IC is the $1-\phi$ fraction). % of this value).

Armed with these tools, we proceed to draw maps of
elections. First, we consider $p$-ID, $p$-IC, and
$(p,\phi)$-resampling elections, where the $p$ and $\phi$ values are
chosen to form a grid, and compute the approvalwise distances between
them. 
% Second, for each election we compute a point on a plane, so that
% the Euclidean distances between these points resemble the approvalwise
% distances between the respective elections as much as possible (we use
% the same algorithm for this purpose as
% \citet{boe-bre-fal-nie-szu:c:compass}). 
We find that, for a fixed value
of $p$, the $(p,\phi)$-resampling elections indeed form lines between
the $p$-ID and $p$-IC ones, whereas for fixed $\phi$ values they form
lines between $0$-ID and $1$-ID ones (which we refer to as the
\emph{empty} and \emph{full} elections).  We obtain more maps by
adding elections generated according to other statistical
cultures; the presence of the \emph{$(p,\phi)$-resampling grid} helps
in understanding the locations of these new elections.
For each of our elections we compute several parameters, such as, e.g,
the highest number of approvals that a candidate receives, the time
required to compute the results of a certain multiwinner voting rule,
or the cohesiveness level (see Section~\ref{sec:prelims} for a
definition). For each of the statistical cultures, we present maps
where we color the elections according to these values. This gives
further insight into the nature of the elections they generate.
Finally, we compare the results for randomly generated elections with
those appearing in real-life, in the context of participatory
budgeting.

\section{Preliminaries}\label{sec:prelims}
For a given positive integer $t$, we write $[t]$ to denote the set
$\{1,2,\dots,t\}$, and $[t]_0$ as an abbreviation for
$[t] \cup \{0\}$.  

%Beta distribution is is defined over interval
%$[0,1]$ and parameterized by two positive real numbers, $\alpha$ and
%$\beta$. Its probability density function is
%$
%f(x;\alpha,\beta) = \textstyle \frac{1}{B(\alpha,\beta)} x^{\alpha-1}(1-x)^{\beta-1}
%$
%where $ B = \frac{\Gamma(a) \Gamma(b)}{\Gamma(a+b)}$, and %$\Gamma$ is the Gamma function.
    
\paragraph{Elections.}

A (simple) approval election $E=(C,V)$ consists of a set of candidates
$C=\{c_1,\dots,c_m\}$ and a collection of voters
$V=(v_1,\dots,v_n)$. Each voter $v \in V$ casts an approval ballot,
i.e., he or she selects a subset of candidates that he or she
approves.  Given a voter $v$, we denote this subset by
$A(v)$. Occasionally, we refer to the voters or their approval
ballots as votes; the exact meaning will always be clear from the
context.
%
% Formally, we identify voters with their binary approval vector: $v=(a_1,\dots,a_m)$ where % $a_i$ (for $i\in [m]$) is 1 if voter $v$ approves candidates $c_i$ and 0 otherwise. Let  $A(v)$
%denote the approval ballot of $v$, i.e., the set of all candidates approved by voter $v$, and let $A^c(v)=C\setminus A(v)$ denote its complement.
%By $\hat{v}=\frac{v}{\sum v}$ we denote the normalized approval ballot,
%where all the ones are divided by the number of approved candidates, 
%so the sum of $\hat{v}$ adds up to $1$.
%
%\begin{example}
%    Let $C=\{c_1,c_2,c_3\}$, and let voter $v$ approves $c_1$ and $c_3$. Then, the approval ballot is $A(v)=\{c_1,c_3\}$, the approval vector is $v=(1,0,1)$, and the normalized %approval ballot is $\hat{v}=(\nicefrac{1}{2},0,\nicefrac{1}{2})$.
%\end{example}
%
An approval-based committee election (an ABC election) is a triple
$(C,V,k)$, where $(C,V)$ is a simple approval election and $k$ is the
size of the desired committee.  We use simple elections when the goal
is to choose a single individual, and ABC elections when we seek a
committee.

% In the single-winner setting, the goal is to choose a single individual and, so,
% it suffices to use simple approval elections.
% In the multiwinner setting we aim to select a committee and, thus, we need to indicate
% its size. 
%
%A specific type of election are multi-winner elections, where the goal is to determine a fixed-size subset of candidates, a so-called committee. These elections additionally require a desired committee size $k$. 
%
% Formally, an approval-based committee election (an ABC election) is a triple $(C,V,k)$,
% where (C,V) is a simple approval election and $k$ is the committee size.

Given an approval election $E$ (be it a simple election or an ABC
election) and a candidate $c$, we write $\score_\av(c)$ to denote the
number of voters that approve $c$. We refer to this value as the
\emph{approval score} of $c$. The single-winner approval rule (called AV)
returns the candidate with the highest approval score (or the set of
such candidates, in case of a tie).

\paragraph{Distances Between Votes.}
%We briefly mention two distance measures between votes~\cite{caragiannis2022evaluating}.

For two voters $v$ and $u$, their Hamming distance is \linebreak
$\ham(v, u) = |A(v) \triangle A(u)|$, i.e.,  the number of
candidates approved by exactly one of them. Other distances include,
e.g., the Jaccard one, defined as
$\mathrm{jac}(v, u) = \frac{\ham(v,u)}{|A(v) \cup A(u)|}$.
For other  examples of such distances, 
we point to the work
of \citet{caragiannis2022evaluating}.

% or the
% Zelinka one,
% $\mathrm{zel}(v, u) = \max\{|A(v)\setminus A(u)|, |A(u)\setminus
% A(v)|\}$.  See the work of \citet{caragiannis2022evaluating} for
% applications of these distances in the context of approval elections.

% \begin{itemize}
%     \item Hamming: $\ham(v, u) = |A(v) \triangle A(u)|$ -- the number of candidates that are approved by exactly one of the voters.
%     \item Jaccard: $jac(v, u) = \frac{\ham(v,u)}{|A(v) \cup A(u)|}$
%     \item Zelinka: $zel(v, u) = \max\{|A(v)\setminus A(u)|, |A(u)\setminus A(v)|\}$
%     \item Bunke-Shearer: $bus(v, u) =  \frac{zel(v,u)}{\max\{|A(v)|, |A(u)|\}}$
% \end{itemize}
% In our work, we mostly use the Hamming distance, but discuss 
% others as well (see Section~\ref{???}).
% \kso{remove it?}

\paragraph{Approval-Based Committee Voting Rules.}

An \emph{approval-based committee voting rule} (an ABC rule) is a function that maps an ABC election $(C,V, k)$ to a nonempty set of committees of size $k$.
If an ABC rule returns more than one committee, then we consider them tied.
%
%for a given ABC election, we say that these committees are \emph{tied}.
% Due to this definition, we assume that ABC rules are resolute, i.e., return exactly one winning committee.
% If an ABC rule is \emph{irresolute}, it instead maps to a set of one or more tied committees.

We introduce two prominent ABC rules. \emph{Multiwinner Approval  Voting (AV)} selects the $k$ candidates with the highest approval scores.
Given a committee $W$, its approval score is the sum of the scores of its members; $\score_\av(W) = \sum_{w \in W}\score_\av(w)$.
If there is more than one committee that achieves a maximum score, AV returns all tied committees.
The second rule is \emph{Proportional Approval Voting (PAV)}. PAV outputs all committees with maximum PAV-score:
\[ \textstyle
  \score_{\text{pav}}(W)=\sum_{v\in V} h(|A(v)\cap W|),
\]
where $h(x)=\sum_{j=1}^x \nicefrac{1}{j}$ is the harmonic function.
Intuitively, AV selects committees that contain the ``best'' candidates (in the sense of most approvals) and PAV selects committees that are in a strong sense proportional \cite{justifiedRepresentation}.
In contrast to AV, which is polynomial-time computable, PAV is NP-hard to compute \cite{azi-gas-gud-mac-mat-wal:c:approval-multiwinner,sko-fal-lan:j:collective}.
In practice, PAV can be computed  by solving an integer linear program \cite{jair/spoc} or by an approximation algorithm~\cite{DudyczMMS20-tight-pav-apx}.

\iffalse
AV belongs to the class of Thiele methods. A Thie\-le method is defined by a scoring function $w$ mapping non-negative integers to reals; it returns all committees $W$ that achieve a maximum $w$-score, defined as: 
\[ \textstyle
  \score_w(W)=\sum_{v\in V} w(|A(v)\cap W|).
\]
Observe that AV is defined by $w(x)=x$.
Two other major Thiele methods are Proportional Approval Voting (PAV) and Chamberlin--Courant (CC), which are defined by \[w_{\mathit{pav}}(x)=\sum_{j=1}^x \nicefrac{1}{j} \quad\text{ and }\quad w_{\mathit{cc}}(x)=\begin{cases}1 & \text{if $x\geq 1$,}\\ 0 & \text{otherwise.}\end{cases}\]
Intuitively, AV selects committees that contain the ``best'' candidates (in the sense of most approvals), PAV selects committees that are in a strong sense proportional \cite{justifiedRepresentation}, and 
CC selects committees that maximally representative (in the sense that most voters approve at least
one candidate).
\fi

\paragraph{Cohesive Groups.}

%From an intuitive point of view
Intuitively, a proportional committee should represent all groups of voters in a way that (roughly) corresponds to their size.
To speak of proportional committees in ABC elections, \citet{justifiedRepresentation} introduced the concept of \emph{cohesive groups}.
\begin{definition}
Consider an ABC election $(C,V,k)$ with $n$ voters and some non-negative integer $\ell$.
A group of voters $V'\subseteq V$ is \emph{$\ell$-cohesive} if
(i) $|V'| \geq \ell\cdot \frac{n}{k}$
and
(ii) $\left|\bigcap_{v \in V'} A(v) \right| \geq  \ell$.
\end{definition}
An $\ell$-cohesive group is large enough to deserve $\ell$ representatives in the committee and is cohesive in the sense that there are $\ell$ candidates that can represent it.
A number of proportionality notions have been proposed based on cohesive groups,
such as 
(extended) justified representation~\cite{justifiedRepresentation}, proportional justified representation~\cite{Sanchez-Fernandez2017Proportional}, proportionality degree~\cite{sko:c:prop-degree}, and others.
For our purposes, it is sufficient to note that all these concepts guarantee cohesive groups different types of representations represented
(see also the  survey of \cite{lac:sko:t:approval-survey} for a comprehensive overview).

% Let $q$ denote the total number of approvals in all the votes (if two voters approved the same candidate, it counts as two approvals). Let $a^i$ denote the number of candidates that are approved by $i$ voters. By $A=(\frac{a^0}{t},\frac{a^1}{t},\dots,\frac{a^n}{t})$ we denote the global-approval-frequency vector, which contains $a^i$ values for $i\in[t]\cup{0}$ normalized by the total number of approvals $q$.

\section{Statistical Cultures for Approval Elections}
\label{sec:statistical-cultures}

In the following, we present several statistical cultures (probabilistic models) for
generating approval elections.  Our input consists of the desired
number of voters $n$ and a set of candidates $C=\{c_1,\dots,c_m\}$.
For models that already exist in the literature, we provide examples of
papers that use them.

% each model we indicate if it already appeared in the literature or
% if it is introduced in this paper.

\paragraph{Resampling, IC, and ID Models.}
Let $p$ and $\phi$ be two numbers in $[0,1]$. In the
\emph{$(p,\phi)$-resampling} model, we first draw a central
ballot~$u$, by choosing $\lfloor p \cdot m\rfloor$ approved candidates
uniformly at random; then, we generate each new vote $v$ by initially
setting $A(v) = A(u)$ and executing the following procedure for every
candidate $c_i \in C$: With probability $1-\phi$, we leave $c_i$'s
approval intact and with probability $\phi$ we resample its value
(i.e., we let $c_i$ be approved with probability $p$). The resampling
model is due to this paper and is one of our basic tools for analyzing
approval elections.  By fixing $\phi = 1$, we get the
\emph{$p$-impartial culture} model ($p$-IC) where each
candidate in each vote is approved with probability $p$; it was used,
e.g., by \citet{bre-fal-kac-nie2019:experimental_ejr} and
\citet{fal-sli-tal:c:vnw}.  By fixing $\phi = 0$, we ensure that all
votes in an election are identical (i.e., approve the same $p$ fraction of the candidates). We refer this model
as \emph{$p$-identity} ($p$-ID).

% \paragraph{Moving Model.} The \emph{$(p,\phi)$-moving} model is a variant
% of the $(p,\phi)$-resampling one, where each time a new vote is
% generated, this new vote replaces the central one.

\paragraph{Disjoint Model.} The \emph{$(p,\phi, g)$-disjoint} model,
where $p$ and~$\phi$ are numbers in $[0,1]$ and $g$ is a non-negative
integer, works as follows: We draw a random partition of $C$ into $g$
sets, $C_1, \ldots, C_g$, and, to generate a vote, we choose
$i \in [g]$ uniformly at random and sample the vote from a
$(p,\phi)$-resampling model with the central vote that approves exactly
the candidates from~$C_i$ (while the central votes are
independent of $p$, we still need this parameter for resampling).

\paragraph{Noise Models.} Let $p$ and $\phi$ be two numbers from
$[0,1]$ and let $d$ be a distance between approval votes. We require
that~$d$ is polynomial-time computable and, for each two approval votes
$u$ and $v$, $d(u,v)$ depends only on  $|A(u)|$,
$|A(v)|$, and $|A(u) \cap A(v)|$;  both distances from
Section~\ref{sec:prelims} have this property. In the
$(p,\phi,d)$-Noise model we first generate a central vote $u$ 
as in the resampling model and, then, each new vote~$v$ is generated with probability
proportional to $\phi^{d(u,v)}$. Such noise models are analogous to
the Mallows model for ordinal elections and were
studied, e.g., by \citet{caragiannis2022evaluating}. In particular,
they gave a sampling procedure for the case of the Hamming
distance. We extend it to arbitrary distances.

\begin{proposition}
  There is a polynomial-time sampling procedure for the
  $(p,\phi,d)$-noise models (as defined above).
\end{proposition}
\begin{proof}
  Let $u$ be the central vote and let $z = |A(u)|$. Consider
  non-negative integers $x$ and $y$ such that $x \leq z$ and
  $y \leq m - z$. The probability of generating a vote $v$ that
  contains $x$ candidates from $A(u)$ and $y$ candidates from
  $C \setminus A(u)$ is proportional to the following value (abusing notation, we write $d(x,y,z)$ to mean the value $d(u,v)$;
  indeed, $d(u,v)$ depends only on $x$, $y$, and $z$):
  \[
    f(x,y) = \textstyle \binom{z}{x} \binom{m-z}{y} \phi^{d(x,y,z) }.
  \]
  Next, let $Z = \sum_{x \in [z]_0, y \in [m-z]_0} f(x,y)$.
  %; note that it can be computed in polynomial time.
%
  To sample a vote, we draw values $x \in [z]$ and $y \in [m-z]$ with
  probability $\frac{f(x,y)}{Z}$ and form the vote as approving $x$
  random members of $A(u)$ and $y$ random members of
  $C \setminus A(u)$.
\end{proof}
\noindent
In the reminder, we only use the noise model with the Hamming distance
and we refer to it as the $(p,\phi)$-noise model.  Note that the roles of $p$
and $\phi$ in this model are similar but not the same as in the
$(p,\phi)$-resampling model (for example, for $\phi = 0$ we get the
$p$-ID model, but for $\phi=1$ get the $0.5$-IC one).

\paragraph{Euclidean Models.}
In the $t$-dimensional Euclidean model, % each agent $x$
%(i.e,
each candidate and each voter is a point %$p(x)$
from $\reals^t$ and a voter~$v$ approves candidate~$c$ if the distance
between their points is at most $r$ (this value is called the radius);
such models were discussed, e.g., in the classical works
of Enelow and Hinich~[\citeyear{enelow1984spatial},\citeyear{enelow1990advances}], and more recently
by~\citet{elk-lac:c:ci-vi--approval},
\citet{elk-fal-las-sko-sli-tal:c:2d-multiwinner},
\citet{bre-fal-kac-nie2019:experimental_ejr}, and
\citet{god-bat-sko-fal:c:2d}.
We consider $t$-dimensional models for $t \in \{1,2\}$, where the
agents' points are distributed uniformly at random on $[0,1]^t$. We
refer to them as 1D-Uniform and 2D-Square models (note that to fully
specify each of them, we also need to indicate the radius value).

\iffalse
\paragraph{Voter/Candidate Range Models.}
In the voter/candidate range models (the VCR models) each agent $x$
(i.e, each candidate and each voter) is a point $p(x)$ in some
Euclidean space $\reals^t$. Further, for each agent also has radius
$r(x)$. A voter $v$ approves a candidate $c$ if the Euclidean distance
between their points, $p(v)$ and $p(c)$, is at most $r(v) + r(c)$.
VCR models are extensions of the classic Euclidean
models~\cite{enelow1984spatial,enelow1990advances} and the CI/VI
models~\cite{elk-lac:c:ci-vi--approval}, introduced recently by
\citet{god-bat-sko-fal:c:2d}.

We consider $t$-dimensional Euclidean spaces where $t$ is either $1$,
$2$, or $3$, and where the coordinates of the agents' points are
either distributed uniformly on $[0,1]$ (we refer to these models as
1D-Uniform, 2D-Square, and 3D-Cube) or come from Gaussian distribution
with mean $0.5$ and standard deviation $0.15$ (we refer to these as
1D-, 2D-, and 3D-Gaussian models).  We draw agents radii from the Beta
distribution (for the uniform models we use $\alpha=1$ and $\beta=6$,
and for the Gaussian ones we use $\alpha=2$ and $\beta=12$; both sets
of parameters give expected value $0.142$, but for the former the
probability density function is decreasing, whereas for the latter it
is more bell-shaped).
%
Note that for the higher-dimensional spaces the distances between the
agents' points get larger; we did not normalize them as we also wanted
to observe this effect in our experiments.
\fi

\paragraph{Truncated Urn Models.}
Let $p$ be a number in $[0,1]$ and let~$\alpha$ be a non-negative real
number (the parameter of contagion).  In the truncated
Pólya-Eggenberger Urn Model~\cite{berg1985paradox} we start with an
urn that contains all $m!$ possible linear orders over the candidate
set. To generate a vote, we (1) draw a random order $r$ from the urn,
(2)~produce an approval vote that consists of $\lceil p\cdot m \rceil$
top candidates according to $r$ (this is the generated vote), and
(3)~return $\alpha m!$ copies of $r$ to the urn.
%
%Note that
For $\alpha = 0$, all votes with $\lceil p\cdot m \rceil$ approved
candidates are equally likely, whereas for large values of
$\alpha$ %it is likely that
all votes are likely to be identical (so the model becomes 
similar to $p$-ID).

\section{Metrics}
Next we describe two (pseudo)metrics used to measure
distances between approval elections. Since we are interested in
distances between randomly generated elections, our metrics are
independent of renaming the candidates and %reordering the
voters.

%We first introduce some additional notation.  
Consider two equally-sized candidate sets $C$ and $D$, and a voter $v$ with a ballot over $C$.
%Let $v$ be some voter
%with an approval ballot 
%over candidate set $C$, and let $D$ be some
%candidate set such that $|C| = |D|$. 
For a bijection
$\sigma \colon C \rightarrow D$, by $\sigma(v)$ we mean a voter with
an approval ballot $A(\sigma(v)) = \{ \sigma(c) \mid c \in C\}$. In
other words, $\sigma(v)$ is the same as $v$, but with the candidates
renamed by $\sigma$. We write $\Pi(C,D)$ to denote the set of all
bijections from $C$ to $D$. For a positive integer $n$, by $S_n$ we
mean the set of all permutations over $[n]$.  Next, we define the
isomorphic Hamming distance (inspired by the metrics of
\citet{fal-sko-sli-szu-tal:c:isomorphism}).

% To describe the two (pseudo)metrics we use throughout the paper we first introduce some additional notation.
% Let $v$ be a voter with an approval ballot over a candidate set $C$ and let $D$ be another candidate set of the same size. 
% For a bijection $\sigma \colon C \rightarrow D$, by $\sigma(v)$ we mean a voter with an approval ballot $A(\sigma(v)) = \{ \sigma(c) \mid c \in C\}$. That is, $\sigma(v)$ is similar to $v$, but with the candidates renamed by $\sigma$. We write $\Pi(C,D)$ to denote the set of all bijections from $C$ to $D$. For a positive integer $n$, we use $S_n$ to denote the set of all permutations over $[n]$.
% %
% We are ready to define the \emph{isomorphic Hamming distance} (inspired by the metrics of~\citet{fal-sko-sli-szu-tal:c:isomorphism}).

\begin{definition}
  Let $E = (C,V)$ and $F = (D,U)$ be two elections, where $|C| = |D|$,
  %$C = \{c_1, \ldots, c_m\}$, $D = \{d_1, \ldots, d_m\}$,
  $V = (v_1, \ldots, v_n)$ and $U = (u_1, \ldots, u_n)$.  The
  \emph{isomorphic Hamming distance} between $E$ and $F$, denoted
  $d_{\hamming}(E,F)$, is defined as:
  \begin{align*}
    % d_{\hamming}&(E,E') =\\ &
   \textstyle\min_{\sigma \in \Pi(C,D)}\min_{\rho \in S_n} \left(\sum_{i=1}^n \hamming(\sigma(v_i), u_{\rho(i)}) \right).
  \end{align*}
\end{definition}
\noindent Intuitively, under the isomorphic Hamming distance we unify
the names of the candidates in both elections and match their voters
to minimize the sum of the resulting Hamming distances. We call this
distance \emph{isomorphic} because its value is zero exactly if the
two elections are identical, up to renaming the candidates and voters.
Computing this distance is $\np$-hard (see also the
related results for approximate graph
isomorphism~\cite{arv-koe-kuh-vas:c:approximate-graph-isomorphism,gro-rat-woe:c:approximate-graph-isomorphism}).

\begin{proposition}\label{hamming-hard}
  Computing the isomorphic Hamming distance between two approval elections is NP-hard.
\end{proposition}

Thus we  compute this distance using a brute-force
algorithm (which is faster than using, e.g., ILP formulations).
%
% Hence, to compute isomorphic Hamming distances, we use an integer
% linear program (ILP), based on the one provided by
% \citet{fal-sko-sli-szu-tal:c:isomorphism} for one of their distances
% (see the appendix for exact formulation). However, solving this ILP is
% challenging for larger elections and, so,
%
Since this limits the size of elections we can deal with, we also
introduce a simple, polynomial-time computable metric.

\newcommand{\fracc}[2]{{#1}/{#2}}
\begin{definition}
  Let $E$ % = (C,V)$ 
  be an %approval 
  election 
  with candidate set $\{c_1, \ldots, c_m\}$ and
%  $\{c_1, \ldots, c_m\}$ and 
$n$ voters. 
  %$V = (v_1, \ldots, v_n)$. 
  %We define
  Its approvalwise vector, denoted $\av(E)$, is
  obtained by sorting the vector
  $(\fracc{\score_\av(c_1)}{n}$, $\ldots, \fracc{\score_\av(c_m)}{n})$
  in non-increasing order.
\end{definition}
\begin{definition}
  The approvalwise distance between elections $E$ and $F$ with approvalwise
  vectors $\av(E) = (x_1, \ldots, x_m)$ and $\av(F) = (y_1, \ldots, y_m)$
  is defined as: $$d_\app(E,F) = |x_1-y_1| + \cdots + |x_m-y_m|.$$
\end{definition}

\noindent In other words, the approvalwise vector of an election 
is a sorted vector of the normalized approval scores of its candidates, and an approvalwise distance between two elections is the
$\ell_1$ distance between their approvalwise vectors.
We sort the vectors % Approvalwise vectors are sorted 
to avoid the explicit use of a candidate matching,
as is needed in the Hamming distance.  Occasionally we will speak of
approvalwise distances between approvalwise vectors, without
referring to the elections that provide them.

It is easy to see that the approvalwise distance is computable in
polynomial time.  Indeed, its definition is so simplistic that it is
natural to even question its usefulness.  However, in Section~\ref{sec:correlation} we
will see that in our election datasets it is strongly correlated with the
Hamming distance. Thus, in the following discussion, we focus on
approvalwise distances.

\begin{figure}[t]
\centering
\begin{minipage}[b]{0.4\textwidth}
    \scriptsize
   \begin{center}
     \begin{tikzpicture}
       \node[inner sep=3pt, anchor=south] (up)  at (1,4) {full};
       \node[inner sep=3pt, anchor=east] (left)  at (0,2) {0.5-IC};
       \node[inner sep=3pt, anchor=west] (right)  at (2,2) {0.5-ID};
       \node[inner sep=3pt, anchor=north] (down)  at (1,0) {empty};
       
        \draw[draw=black, line width=0.4mm, -]     (left) edge node  {} (up);
        \draw[draw=black, line width=0.4mm, -]     (left) edge node  {} (right);
        \draw[draw=black, line width=0.4mm, -]     (left) edge node  {} (down);
        \draw[draw=black, line width=0.4mm, -]     (up) edge node  {} (right);
        \draw[draw=black, line width=0.4mm, -]     (up) edge node  {} (down);
        \draw[draw=black, line width=0.4mm, -]     (right) edge node  {} (down);
        
        \normalsize
        \node[] at (0,3.4) {$\nicefrac{m}{2}$};
        \node[] at (2,3.4) {$\nicefrac{m}{2}$};
        \node[] at (0,0.6) {$\nicefrac{m}{2}$};
        \node[] at (2,0.6) {$\nicefrac{m}{2}$}; 
        
        \node[] at (1.5,2+0.25) {$\nicefrac{m}{2}$}; 
        \scriptsize
        \node[] at (1+0.25,3) {$m$}; 
        
        \scriptsize
        % \color{blue}
        \node[] at (0.05-0.4,1.25) {$p$-IC};
        \node[] (pic) at (0.05,1.25) [circle,fill,inner sep=1.5pt]{};
        \node[] at (1.95+0.4,1.25) {$p$-ID}; 
        \node[] (pid) at (1.95,1.25) [circle,fill,inner sep=1.5pt]{};
        \draw[draw=black, line width=0.2mm, -]     (pic) edge node  {} (pid);
        \node[] at (0.66,1.05) {$(p,\phi)$}; 
        \node[] (p) at (0.66,1.25) [circle,fill,inner sep=1.5pt]{};
        
        \node[] at (1.4,1.35) {$a$};
        \node[] at (0.4,1.38) {$b$}; 
        \node[] at (0.67,0.62) {$c$}; 
        \node[] at (0.74,3.0) {$d$};  
        \node[anchor=west] at (-0.1,5.6) {$a=2mp(1-p)\phi$};
        \node[anchor=west] at (-0.1,5.3) {$b=2mp(1-p)(1-\phi)$};
        \node[anchor=west] at (-0.1,5.0) {$c=mp$};  
        \node[anchor=west] at (-0.1,4.7) {$d=m(1-p)$};  
        
        \draw[draw=black, line width=0.2mm, -]     (p) edge node  {} (up);
        \draw[draw=black, line width=0.2mm, -]     (p) edge node  {} (down);
    
      \end{tikzpicture}
    \end{center}
    \caption{Distances between resampling elections.}
    \label{compass}
  \end{minipage}%
  \;\;
\begin{minipage}[b]{0.4\textwidth}
    \centering
    \includegraphics[width=3.9cm]{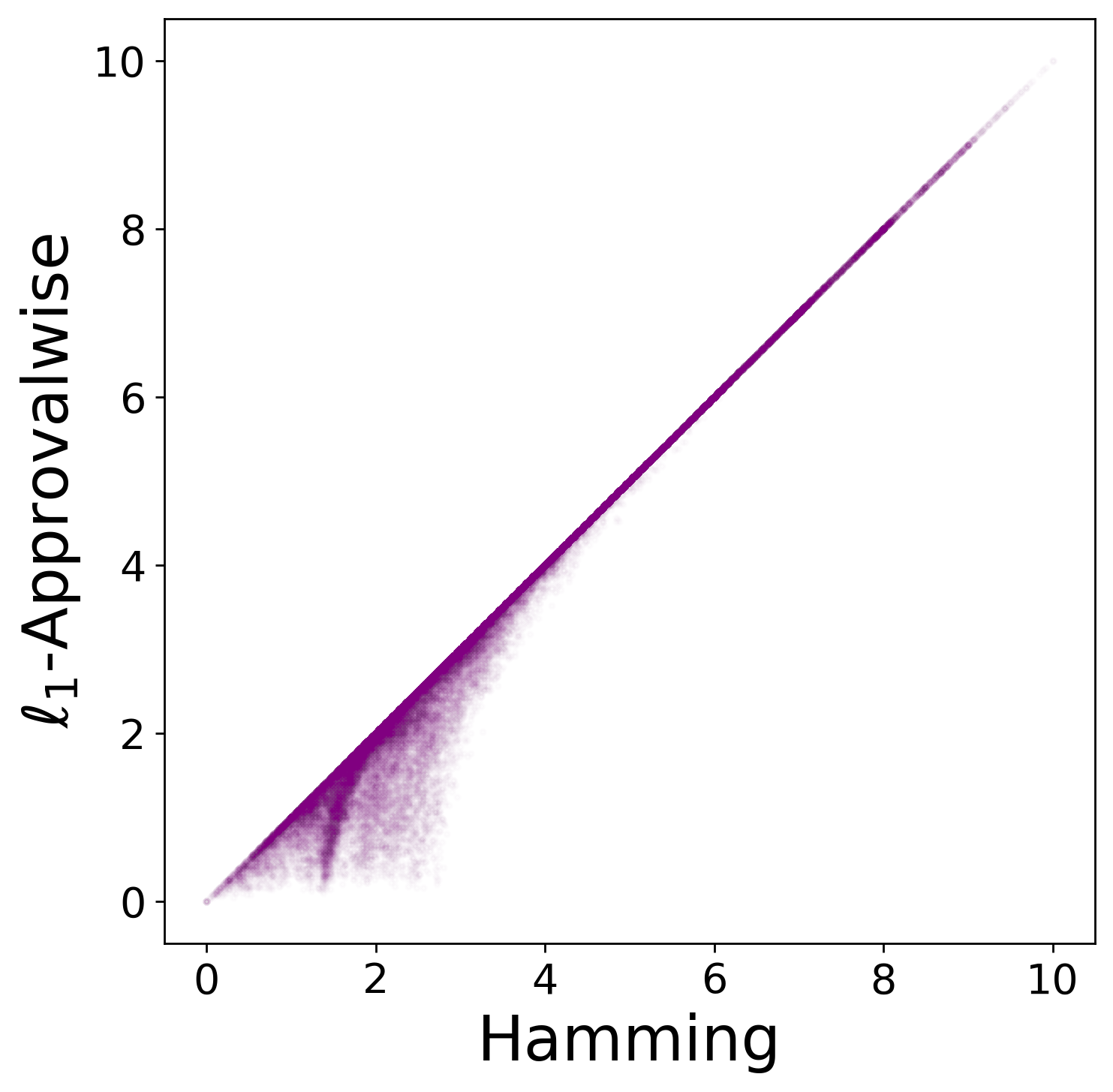}%
    \caption{Correlation between isomorphic Hamming and approvalwise metrics.}
    \label{correlation}
\end{minipage}%
% \begin{minipage}[b]{0.5\textwidth}
%     \includegraphics[width=4.2cm, trim={3 2 3 2}, clip]{images/final_egg/resampling.png}%
%     \includegraphics[width=4.2cm]{images/merged/resampling.jpg}%
%     \caption{Resampling Model.}
%     \label{resampling_main}
% \end{minipage}%
\end{figure}

\newcommand\width{3.3cm}
\newcommand\smallwidth{2.2cm}
\newcommand\finalwidth{5.2cm}

\begin{figure*}[]
    \centering
    % \scalebox{0.95}{
    \includegraphics[width=\finalwidth, trim={3 0 3 0}, clip]{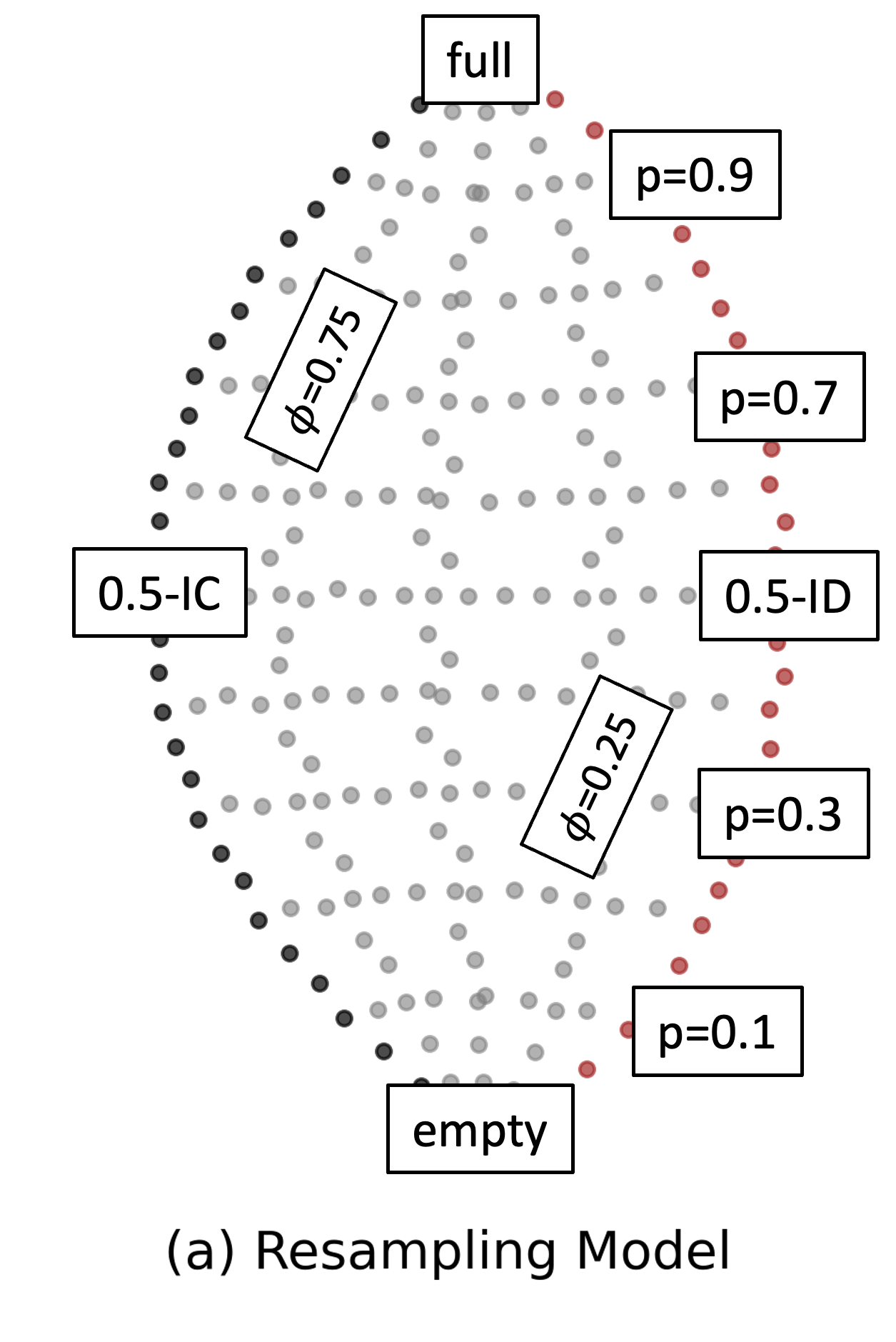}%
    \includegraphics[width=\finalwidth]{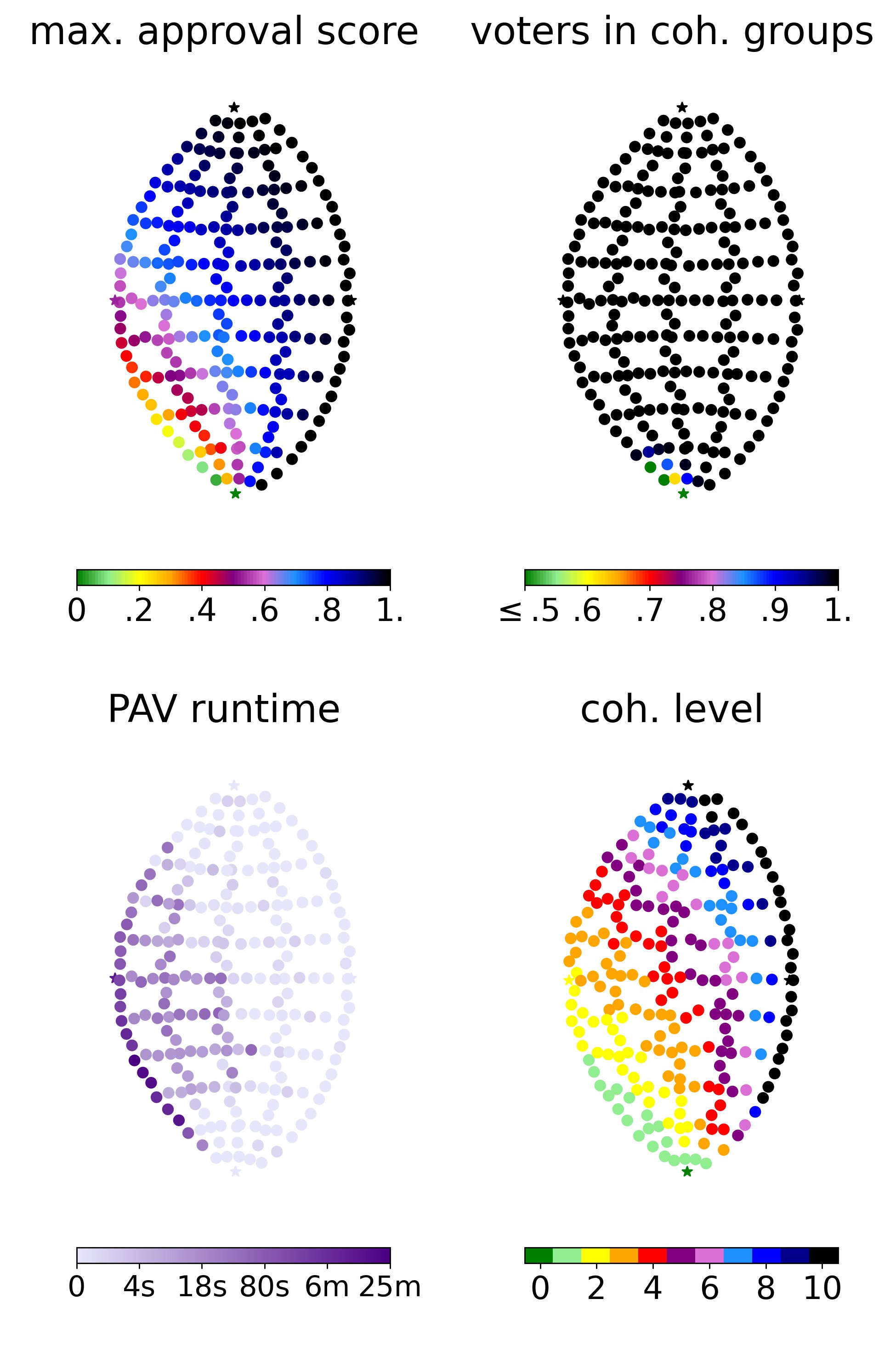}
    
    \includegraphics[width=\finalwidth, trim={3 0 3 0}, clip]{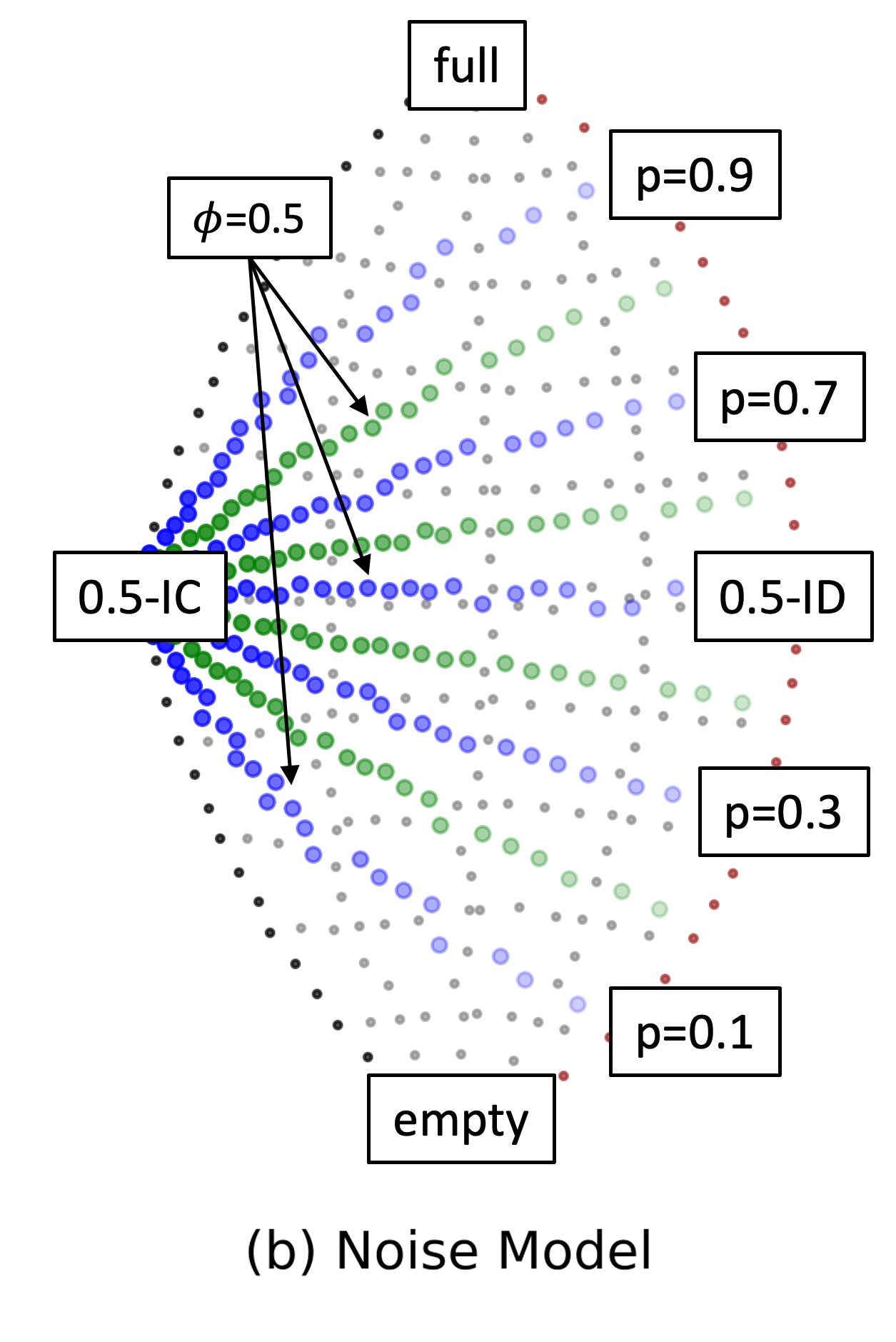}%
    \includegraphics[width=\finalwidth]{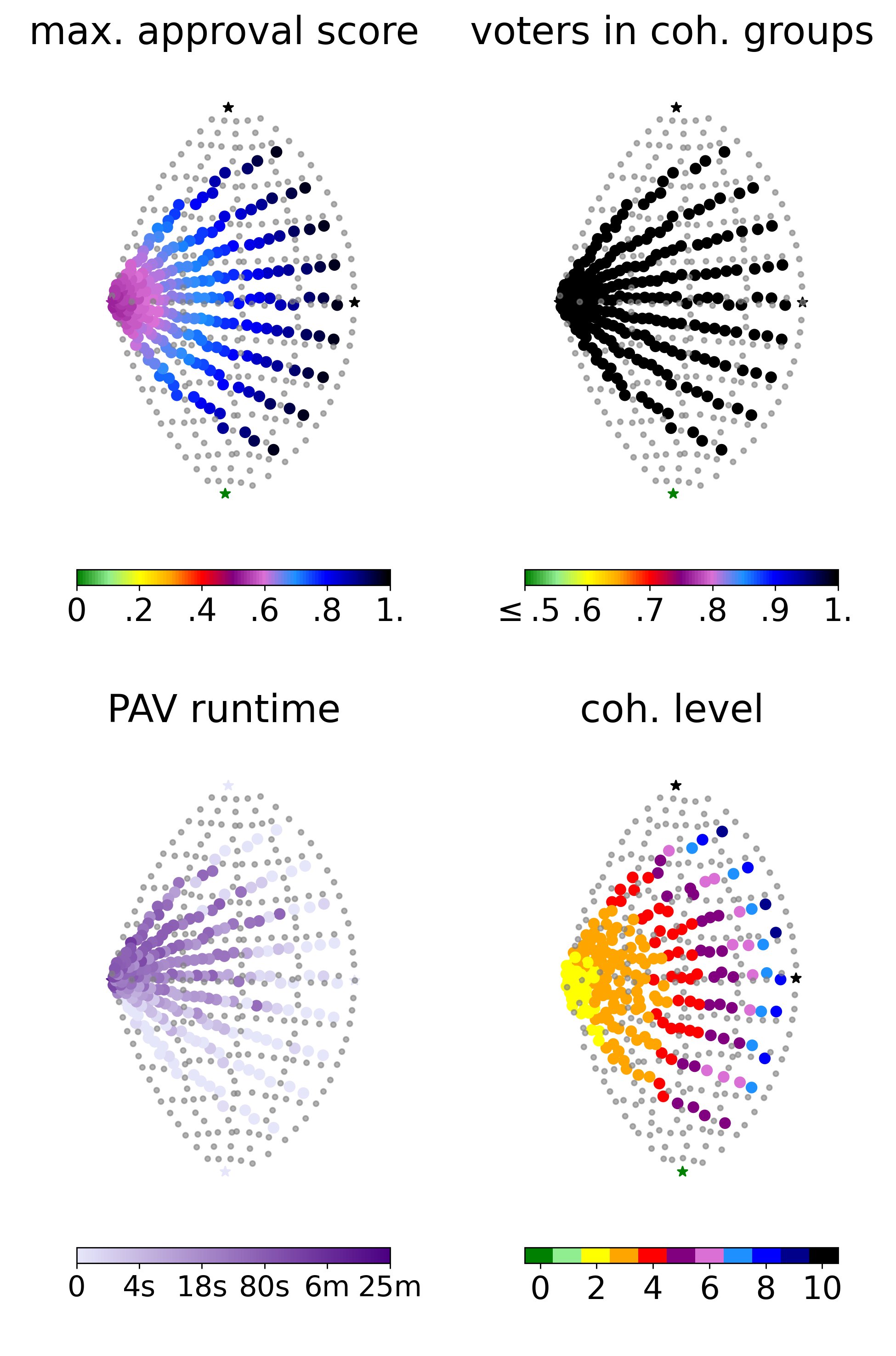}%
    % }
    \caption{Maps for (a)~the Resampling Model and (b) the Noise
      Model. For the latter, the darker a dot in the main plot, the
      higher is the $\phi$ value used to generate the election.}
     \label{main_results_1}
\end{figure*}

\section{A Grid of Approval Elections}\label{sec:grid}

To better understand the approvalwise metric space of elections, next we analyze expected distances between elections generated according
to the $(p,\phi)$-resampling model.

Fix some number $m$ of candidates and parameters
$p, \phi \in [0,1]$, such that $pm$ is an
integer, and consider the process of generating votes from
the $(p,\phi)$-resampling model. 
% As we generate more and more votes,
In the limit, the approvalwise vector of the resulting election is:
\[
  (\underbrace{(1-\phi) + (\phi \cdot p), \ldots, (1-\phi) +
    (\phi\cdot p)}_{p\cdot m}, \underbrace{\phi\cdot p, \ldots,
    \phi\cdot p}_{(1-p)\cdot m}).
\]
Indeed, each of the $p m$ candidates approved in the central
ballot either stays approved (with probability $1-\phi$) or is
resampled (with probability $\phi$, and then gets an approval with
probability $p$). Analogous reasoning applies to the remaining $(1-p)m$
candidates.
With a slight abuse of notation, we call the above vector
$\av(p,\phi)$. Furthermore, we refer to $\av(p,0)$ as the $p$-ID vector,
to $\av(p,1)$ as the $p$-IC vector, and to $0$-ID and \linebreak $1$-ID vectors
as the \emph{empty} and \emph{full} ones, respectively (note that 0-ID $=$ 0-IC and \linebreak 1-ID $=$ 1-IC).

Now, consider two additional numbers, $p', \phi' \in [0,1]$, such that
%$p' \leq p$, $\phi' \leq \phi$, and 
$p' m$ is an integer. Simple
calculations show that:
\begin{align*}
  d_\app(\mathit{empty},\mathit{full}) &= m, \\
  d_\app(p\hbox{-}\mathrm{IC}, p\hbox{-}\mathrm{ID}) &= 2mp(1-p),\\
  d_\app(\av(p,\phi), \av(p',\phi) ) &= m \cdot|p-p'|,\\
  d_\app(\av(p,\phi), \av(p,\phi') ) &= 2mp(1-p)\cdot|\phi-\phi'|.
\end{align*}
Thus $d_\app(\av(p,\phi), \mathit{empty})=mp$ is a $p$~fraction of the
distance between \emph{empty} and \emph{full}, and
$d_\app(\av(p,\phi), p\hbox{-}\mathrm{ID}) = 2mp(1-p)\phi$ is a $\phi$
fraction of the distance between $p$-IC and $p$-ID (see also
Figure~\ref{compass}).  Furthermore,
$d_\app(\mathit{empty},\mathit{full}) = m$ is the largest possible
approvalwise distance.

Intuitively, $(p,\phi)$-resampling elections form a grid that spans
the space between extreme points of our election space; the larger
the~$\phi$ parameter, the more ``chaotic'' an election becomes
(formally, the closer it is to the $p$-IC elections), and the larger
the $p$~parameter, the more approvals it contains (the closer it is to the
\emph{full} election). We use $(p,\phi)$-resampling elections as
a \emph{background} dataset, which consists of $241$ elections with 100
candidates and 1000 voters each, with the following $p$ and $\phi$
parameters:
\begin{enumerate}
\item $p$ is chosen from $\{0, 0.1, 0.2, \dots, 0.9,1\}$ and $\phi$ is
  chosen from the interval $(0,1)$,\footnote{By generating $t$
    elections with a parameter from interval $(a,b)$, we mean
    generating one election for each value $a+i\frac{b-a}{t+1}$, for
    $i \in [t]$.}  
    % or
\item $\phi$ is chosen from $\{0, 0.25, 0.5, 0.75, 1\}$ and $p$ is
  chosen from the interval $(0,1)$.
\end{enumerate}
For each of these elections we compute a point in $\reals^2$, 
% We compute the approvalwise distances between all these elections and
% associate each election with a point in $\reals^2$,
so that the Euclidean distances between these points are as similar to
the approvalwise distances between the respective elections as
possible.  For this purpose, similarly to \citet{szu-fal-sko-sli-tal:c:map}, we use the Fruchterman-Reingold
force-directed algorithm~\cite{fruchterman1991graph}.  For
the resulting map, see the clear grid-like shape at the left side of Figure~\ref{main_results_1}.\footnote{While our visualizations fit nicely into the two-dimensional
embedding, our election space has a much higher dimension.}
Whenever we present maps of elections later in the paper, we compute
them in the same way as described above (but for datasets that include
other elections in addition to the background ones).

\section{Experiments}

In this section, we use the map-of-elections approach
%
%% showcase our map approach described in Section~\ref{sec:grid}
%% and
% a general method for visualizing the landscape of approval elections.  
% Given a set of elections, we create
% a map by first computing pairwise $\ell_1$-approvalwise distances, and
% then embedding all the points (elections) in the 2-dimensional
% Euclidean space using the Fruchterman-Reingold force-directed
% algorithm \cite{???}.
% use it
to analyze quantitative properties of approval elections generated
according to our models.
%from Section~\ref{sec:statistical-cultures}.
% With our approach, one can visualize how quantitative properties
% relate to different statistical distributions.
%Two main insights are to be gained: First, these visualizations reveal
%
In particular, we will see how an election's position in the grid
influences each of the properties, and what parameters to use to
%
%Second, one can determine how to randomly
%
generate elections with the quantitative property in a desired range.

\subsection{Experimental Design}

Concretely, we consider the following four statistics:

\begin{description}
%  \noindent\textbf
  \item[Max. Approval Score.] The highest approval score
  among all candidates in a given election, normalized by the maximum
  possible score, i.e., the number of voters.

  % \noindent\textbf
  \item[Cohesiveness Level.] The largest integer $\ell$ such that
    there exists an $\ell$-cohesive group (for committee size $10$).
    % This number is also  normalized by the largest possible value, i.e., the committee  size.
%     we
%    assume a committee size of $10$.

  % \noindent\textbf
  \item[Voters in Cohesive Groups.] Fraction of voters that belong
    to at least one 1-cohesive group (committee size $10$).

  % \noindent\textbf
  \item[PAV Runtime.] Runtime (in seconds) required to
  compute a winning committee under the PAV rule, by solving an
  integer linear program provided by the \texttt{abcvoting} library
  \cite{abcvoting}, using the Gurobi ILP solver.
\end{description}

% \begin{figure}
%     \centering
%     \includegraphics[width=6cm]{images/merged/resampling.jpg}
%     \caption{Resampling Model}
%     \label{resampling_features}
% \end{figure}

%%%%%%

% \subsection{Datasets}
% Let us first describe our so-called background dataset, which sketches the shape of our maps and
% provides something akin to survey points (Figure~\ref{resampling_main}, left-hand side).
% The background dataset consists of the following elections with 100 candidates and 1000 voters generated via the $(p,\phi)$-resampling model. The nine horizontal paths correspond to $p \in \{0.1, 0.2, \dots, 0.9\}$ and $\phi \in (0,1)$; the five vertical paths correspond to $\phi \in \{0, 0.25, 0.5, 0.75, 1\}$ and $p \in (0,1)$.

% Our experimental analysis is based on
We use six datasets.  Five of them
are generated using our  statistical cultures
% described in Section~\ref{sec:statistical-cultures}
and consist of 100 candidates
and 1000 voters (except for the experiments related to the
cohesiveness level, where we have 50 candidates and 100 voters, due to
computation time). We have: 250 elections from the Disjoint Model (50 for each
$g \in \{2,3,4,5,6\}$ with $\phi \in (0.05, \nicefrac{1}{g})$);
%225
%elections from Moving Model (25 for each
%$p \in \{0.1, 0.2, \dots, 0.9\}$ with $\phi \in (0,0.01)$);
225 elections from the Noise Model with Hamming distance (25 for each
$p \in \{0.1, 0.2, \dots, 0.9\}$ with $\phi \in (0,1)$); 225 elections
from the Truncated Urn Model (25 for each
$p \in \{0.1, 0.2, \dots, 0.9\}$ with $\alpha \in (0,1)$); 200
elections from Euclidean Model (100 for 1D-Uniform, with radius in
$(0.0025,0.25)$, and 100 for 2D-Square, with radius in
$(0.005, 0.5)$); these parameters are as used by
\citet{bre-fal-kac-nie2019:experimental_ejr}. The sixth dataset uses
real-life participatory budgeting data and contains 44 elections from
Pabulib~\cite{pabulib}, where for each (large enough) election we
randomly selected a subset of 50 candidates and 1000 voters.
%(other
%real-life datasets we considered had much fewer candidates).

\subsection{Experimental Results}

Our visualizations are shown in Figures~\ref{main_results_1} and~\ref{main_results_2}.
%
%When looking at the map visualizations of the background dataset (left-hand side of Figure~\ref{main_results_1}a), we see that this dataset indeed establishes a grid structure,
%which
%
We use the grid structure of the background dataset for comparison
with other datasets.  Notably, some of them do not fill this grid: the
disjoint model (Figure~\ref{main_results_2}b) is restricted to the
lower half (i.e., the disjoint model does not yield elections with
very many approvals), the Euclidean model
(Figure~\ref{main_results_2}a) is restricted to the left half (due to
the uniform distribution of points, its elections are rather ``chaotic''),
and the real-world dataset Pabulib (Figure~\ref{main_results_2}d) is
placed very distinctly in the bottom left part.

To get an intuitive understanding of the four statistics, let us
consider the background dataset in Figure~\ref{main_results_1}a.  We
see that the highest approval score is lowest in the lower left side
and increases towards up and right. This is sensible: If the average
number of approved candidates increases, so does this statistic; also,
if voters become more homogeneous, high-scoring candidates are likely
to exist.  Moreover, regarding voters in cohesive groups, it turns out
that in most elections almost all voters belong to a 1-cohesive group,
with the left lower part as an exception (where there are not enough
approvals to form $1$-cohesive groups). The time needed to find a
winning committee under PAV is correlated with the distance from
0.5-IC. We see that it takes the longest to find winning committees if the election is unstructured.  Similar to the highest approval score,
the cohesiveness level increases when moving up or right in the
diagram. Cohesive groups with levels close to the committee size only
exist in very homogeneous elections (rightmost path) and elections
with many approvals (top part).

We move on to the results for the five other datasets.
% displayed in Figure~\ref{main_results_1}b and Figure~\ref{main_results_2}.
Note that each figure also contains the background dataset (gray dots)
for reference.  These results help to understand the differences
between our statistical cultures.

The maximum approval score statistic provides an insight into whether there is a 
candidate that is universally supported. Instances with a value close to 1 possess such
a candidate. In a single-winner election, this candidate is likely to be 
a clear winner. This is undesirable when simulating contested elections or
shortlisting. Also note that in the real-world
data set (Pabulib) we do not observe the existence of such a candidate.

When looking at the PAV runtime, we find some
% can identify
statistical cultures that generate computationally difficult
elections, e.g., the $(p,\phi)$-resampling model with parameter values
close to $p=0.5$ and $\phi=1$ (0.5-IC), the noise model with
parameters $p\in[0.5,0.9]$ and $\phi>0.5$, and the disjoint model
with $g=2$.  Yet, %It is particularly interesting that
instances
from the real-world dataset, as well as from the Euclidean and urn ones, can be computed very quickly.\footnote{Less than 1 second on a single core (Intel Xeon Platinum 8280 CPU @ 2.70GH)
  of a 224 core machine with 6TB RAM. In contrast, the worst-case instance (0.3-IC) required 25 minutes on 13 cores.}

% the same holds for
% %, and so are
% %the instances from
% the Euclidean models % (1D-Interval and 2D-Square)
% and
% Truncated urn model.\footnote{Less than 1 second on a single core (Intel Xeon Platinum 8280 CPU @ 2.70GH)
% of a 224 core machine with 6TB RAM. In contrast, the worst-case instance (0.3-IC) required 25 min on 13 cores.}

%Let us now consider the two statistics related to cohesive groups.
Concerning voters in cohesive groups, whenever this statistic is close to 1,
it is easy to satisfy most voters with at least one approved candidate
in the committee; such committees are easy to find~\cite{justifiedRepresentation}.
Since many proportional rules take special care of voters that belong
to cohesive groups, in such elections there are no voters that are at
a systematic disadvantage.  In many of our generated elections (almost)
all voters belong to $1$-cohesive groups, but this is not the case for
the real-world, Pabulib data.
Indeed, to simulate Pabulib data well, we would likely need to
provide some new
statistical culture(s).
%
% %have this property.
% %However, note that for our real-world
% %data set (Pabulib), this is absolutely not the case.
% Thus, if we seek to generate instances resembling the Pabulib ones, we
% want instances with a value of this statistic between $0$ and
% $0.6$. %Our experiments show that
% We can achieve this, e.g., with
% $(p,\phi)$-resampling for $p<0.5$ and $\phi<0.2$.
% % as well as the disjoint model with $g=2$.
For the cohesiveness level, we see that all models generate a full
spectrum (i.e., $[0,10]$) of cohesiveness levels.
That said, we expect realistic elections to appear in the ``lower
left'' part of our grid (with few approvals), and such elections tend
to have low cohesiveness levels.  Indeed, this is also the case for
Pabulib elections; hence it is important how proportional rules treat
$\ell$-cohesive groups with small $\ell$.

%As before, Pabulib shows a different picture with only small
%values ($\leq 2$). We conclude that in this dataset there are no large
%cohesive groups.

\subsection{Correlation}\label{sec:correlation}
% As described in Section~\ref{sec:grid},
Figures~\ref{main_results_1} and~\ref{main_results_2} are based on the
approvalwise distance.  We %would like to
argue that they
%these figures
would not change much if we used the (computationally intractable)
isomorphic Hamming distance.  To this end, we generated 363 elections
with 10 candidates and 50 voters from the statistical cultures used in
the previous experiment (detailed dataset description is in the appendix~\ref{apdx:correlation}). We compare Hamming and approvalwise
distances; the results are presented in Figure~\ref{correlation}. Each
dot there represent a pair of elections and its coordinates are the
distances between them, according to the Hamming and approvalwise
metrics. The Pearson Correlation Coefficient is 0.989, and for 67\% of
elections the distances are identical.

\section{Future Work}

An important task for future work is to broadly study real-world datasets with the methods
proposed in this paper. Most of these come from political elections with few candidates.
As our analysis may be influenced by the number of candidates, a direct comparison with the 
figures in this paper is not possible; instead one has to rerun the experiments for 
a similar number of candidates.

\begin{figure*}[t]
    \centering

    % \scalebox{0.95}{
    \includegraphics[width=\finalwidth, trim={3 0 3 0}, clip]{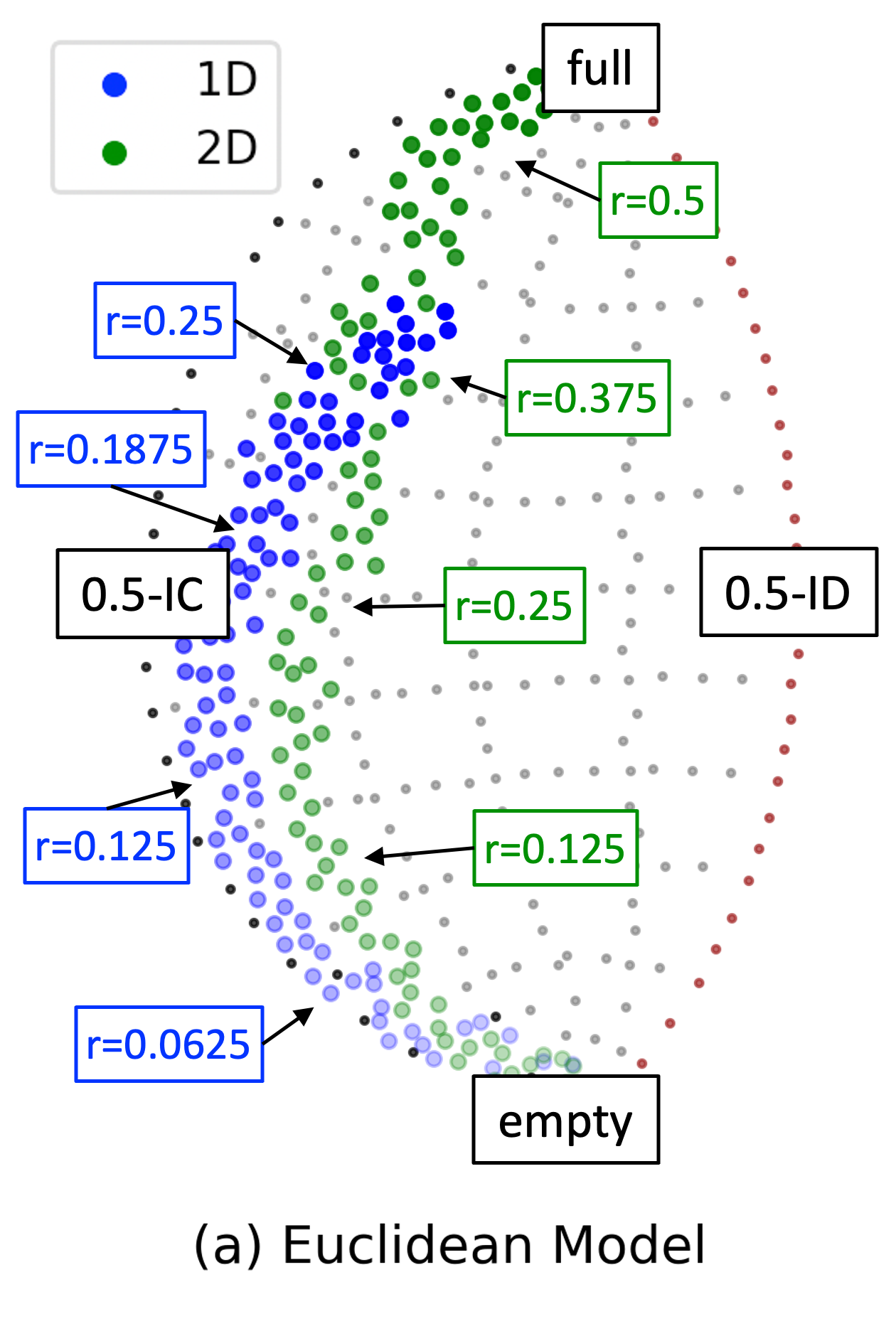}%
    \includegraphics[width=\finalwidth]{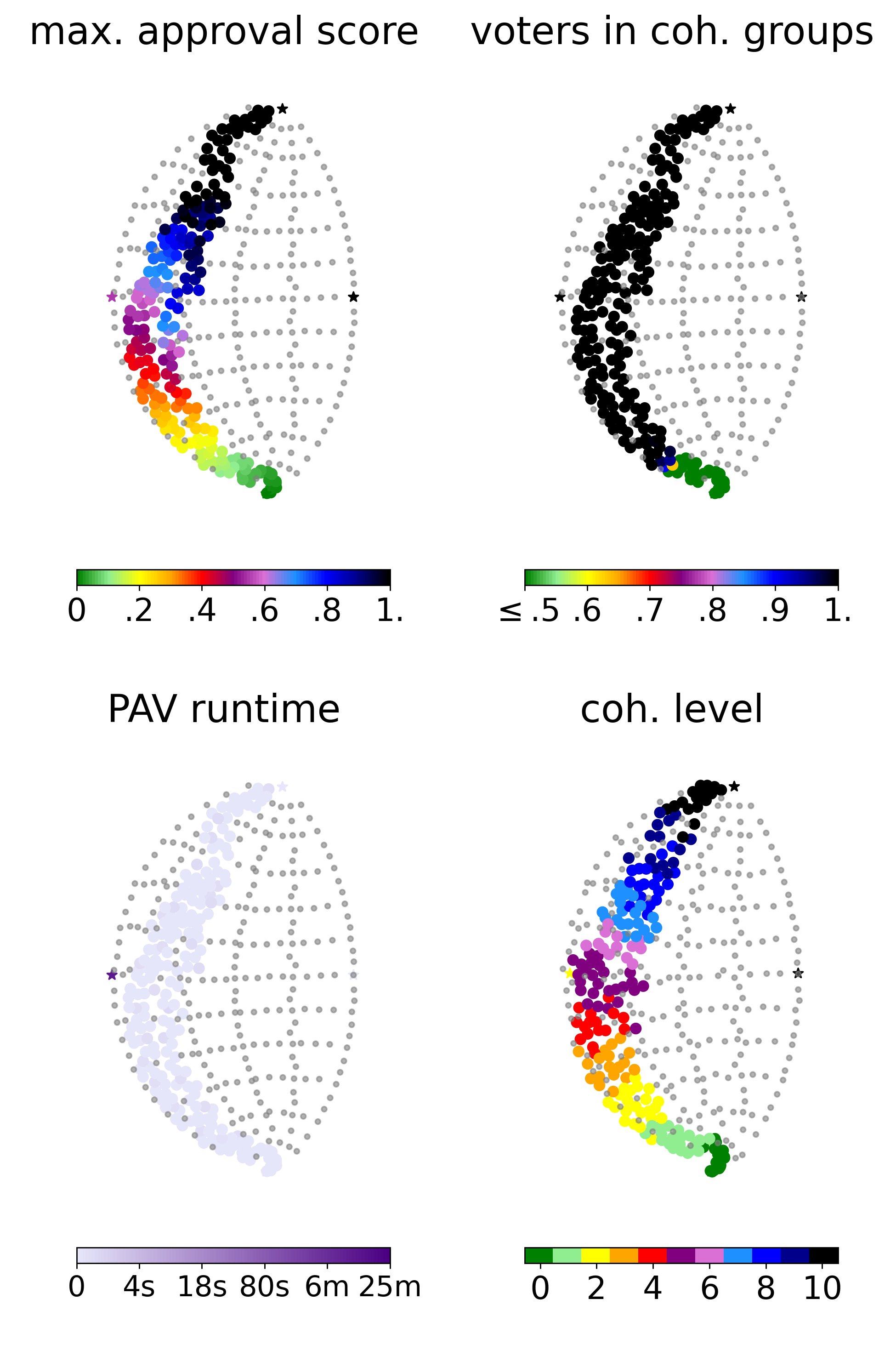}
    \\
    \includegraphics[width=\finalwidth, trim={3 0 3 0}, clip]{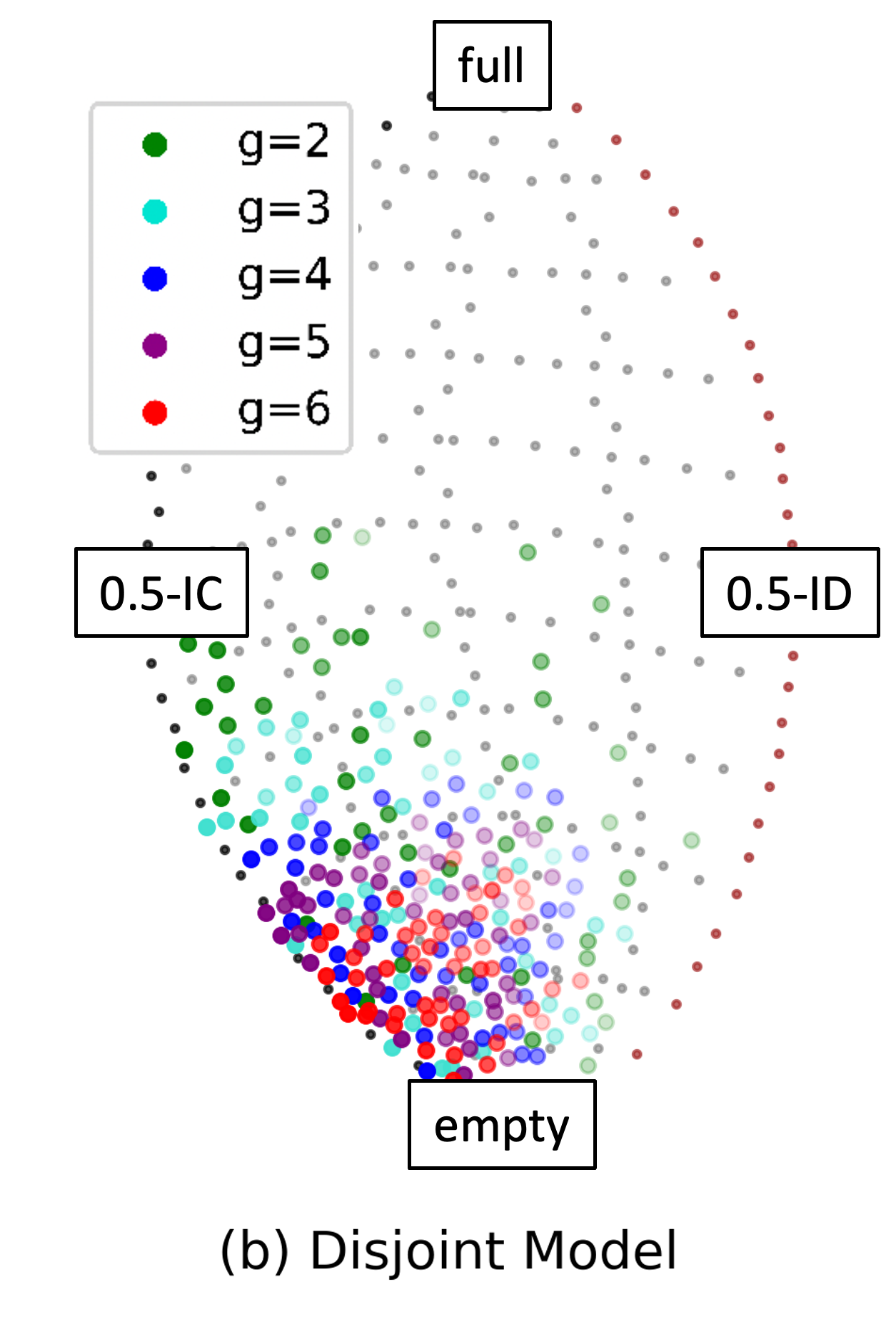}%
    \includegraphics[width=\finalwidth]{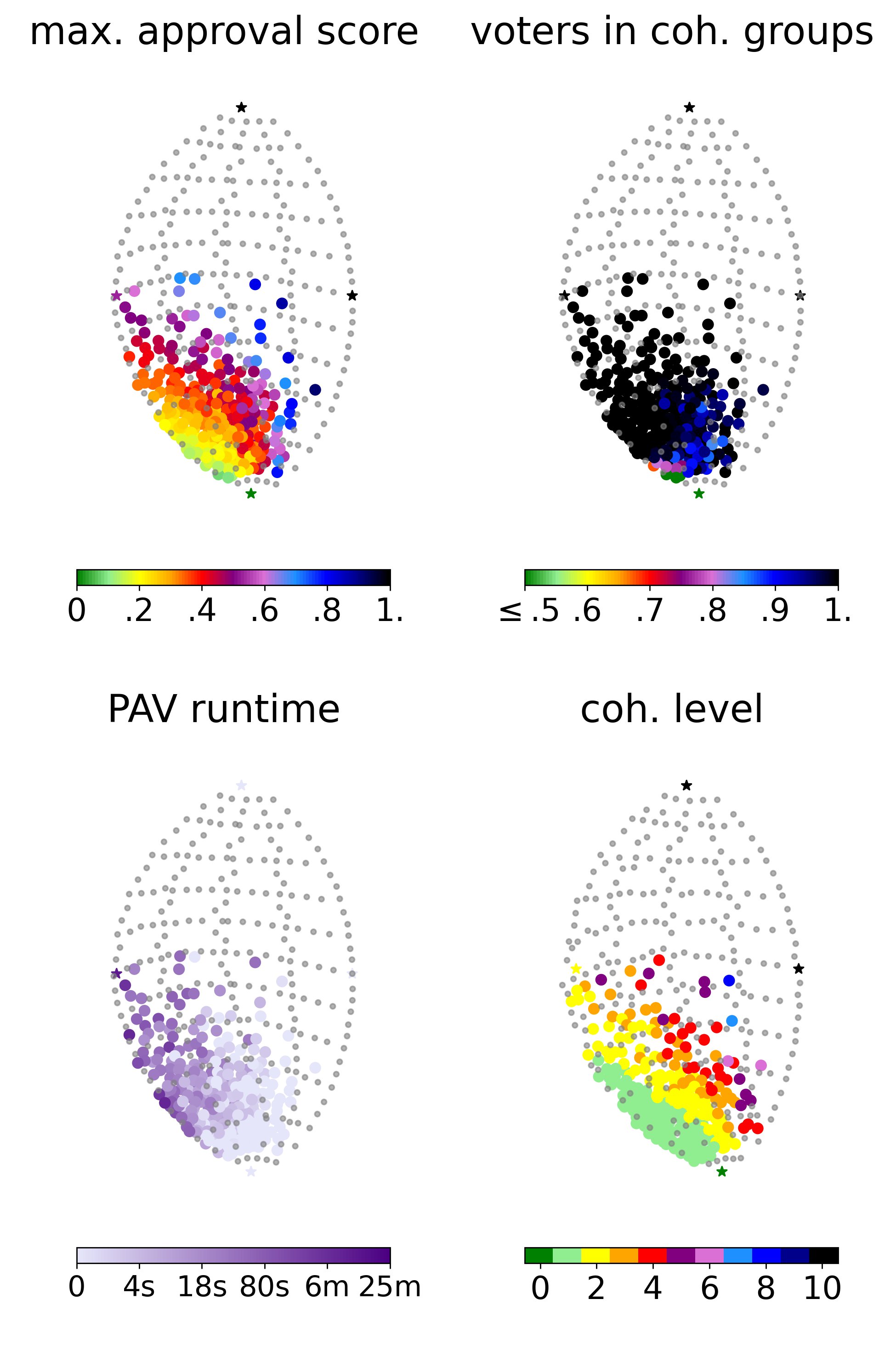}%
%   }
\end{figure*}
\begin{figure*}[t]
    \ContinuedFloat 
    \centering
    % \scalebox{0.95}{
    \includegraphics[width=\finalwidth, trim={3 0 3 0}, clip]{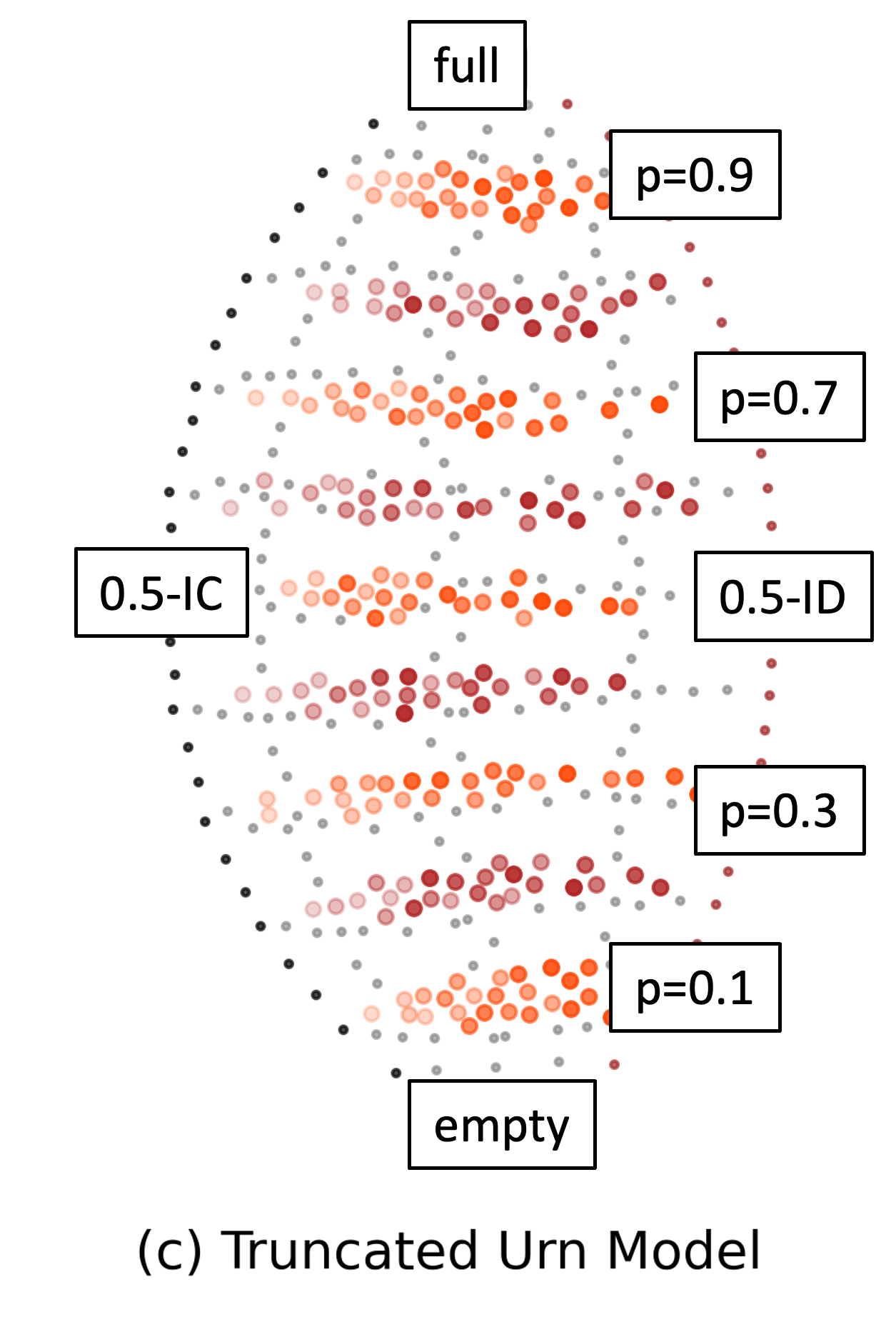}%
    \includegraphics[width=\finalwidth]{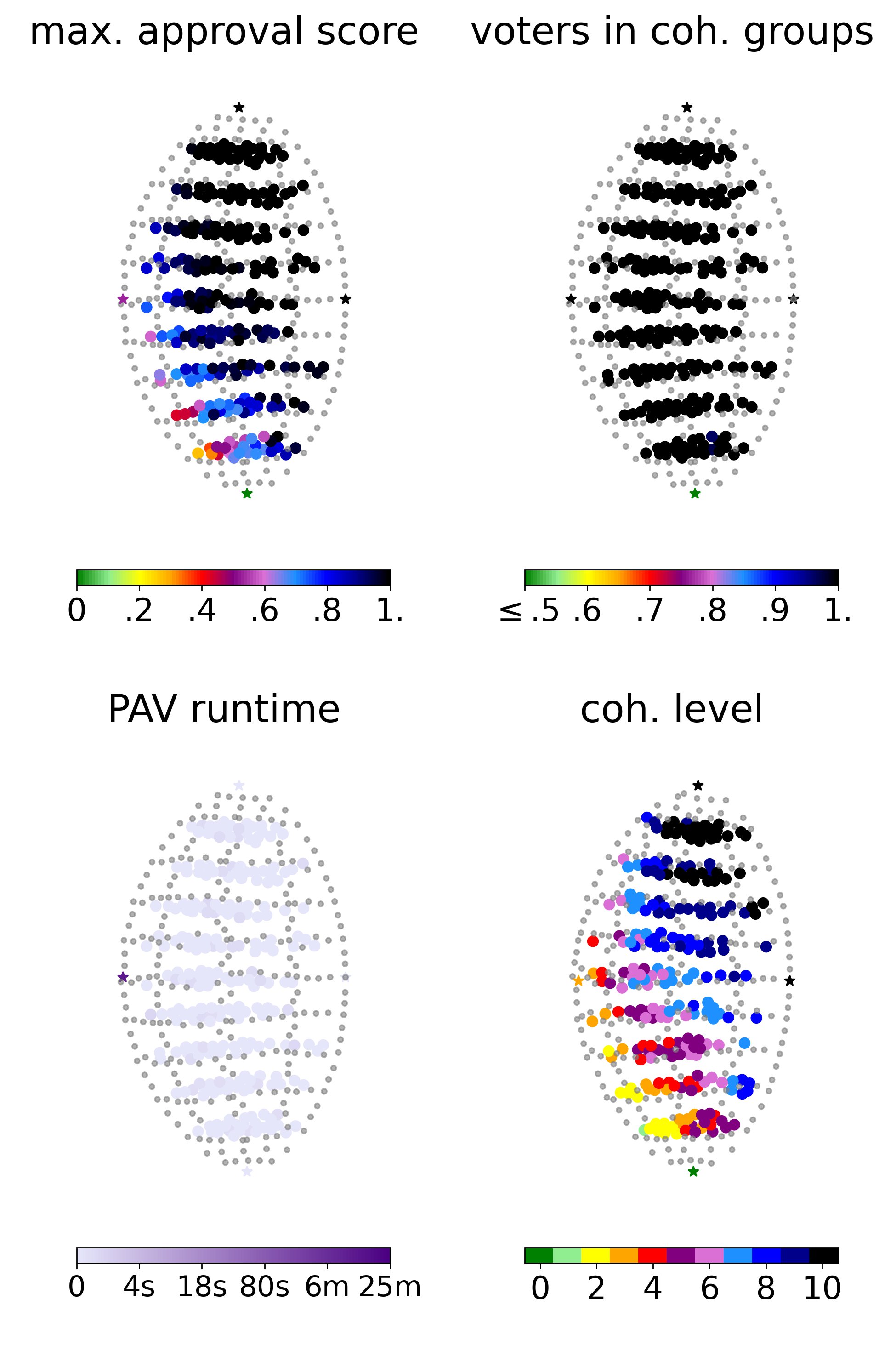}
    \\
    \includegraphics[width=\finalwidth, trim={3 0 3 0}, clip]{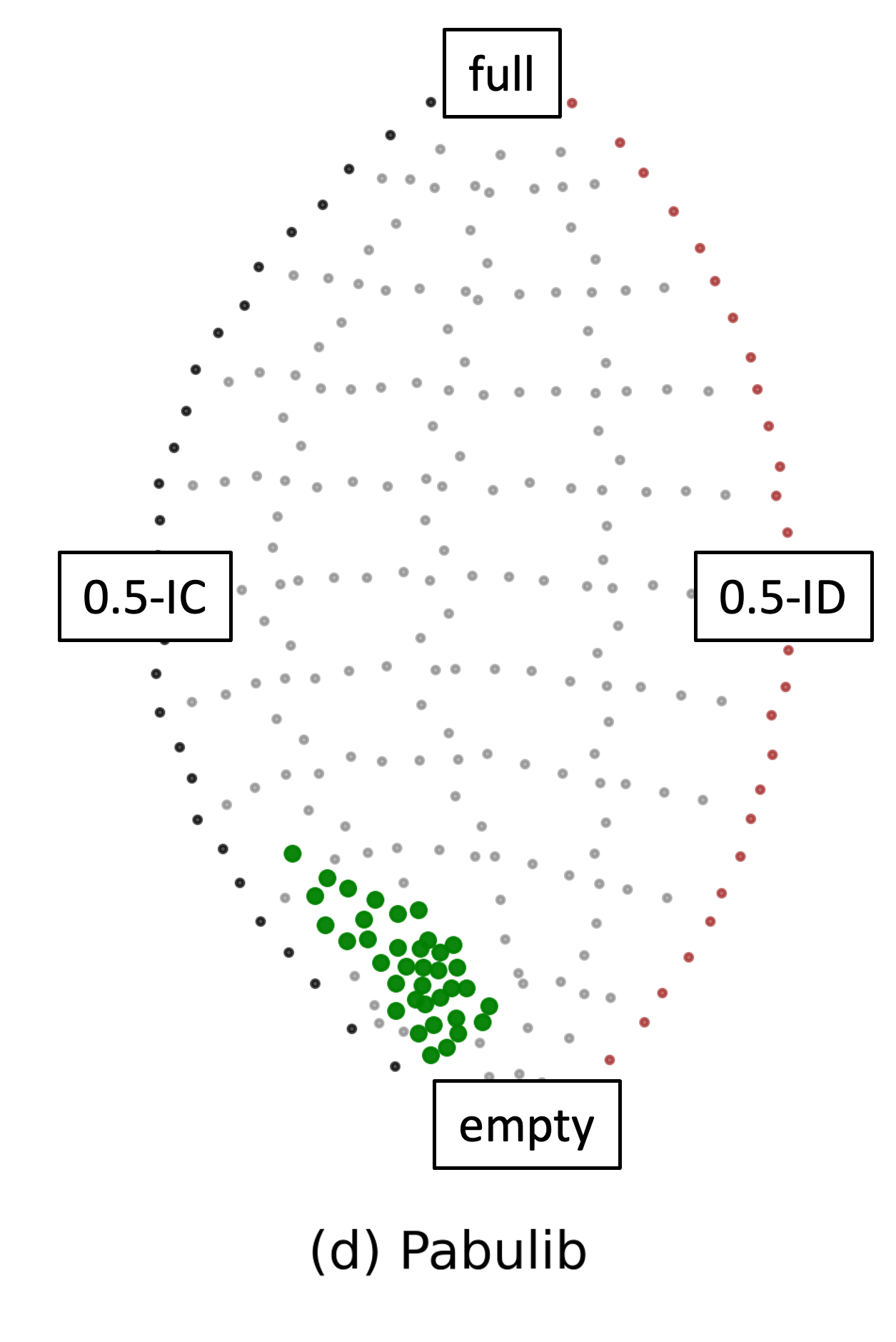}%
    \includegraphics[width=\finalwidth]{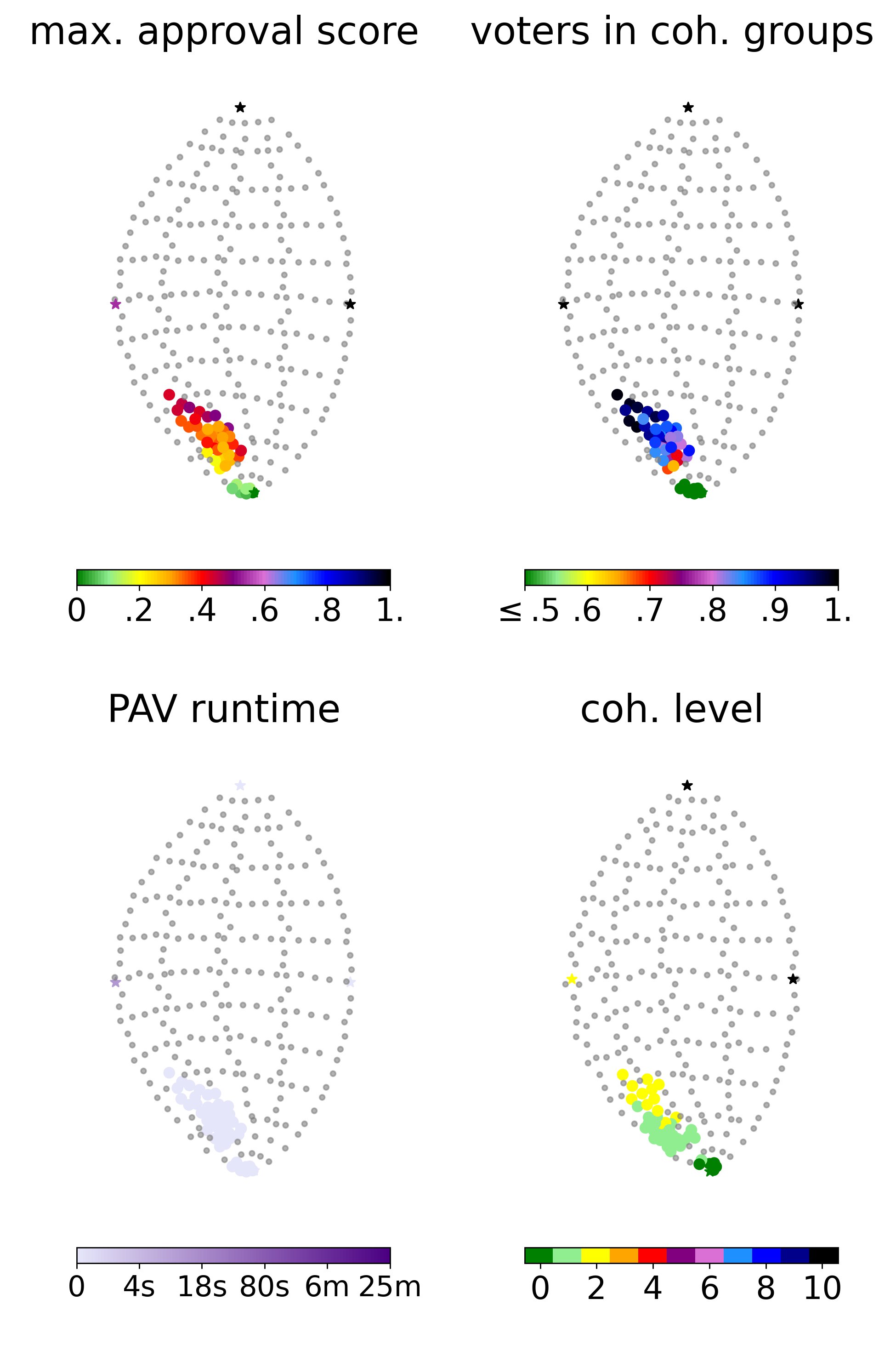}%
    % }
    
    \caption{Maps for (a)~the Euclidean Model, (b) the Disjoint
      Models, (c)~the Truncated Urn Model, and (d)~Pabulib.  For (a),
      (b), and (c), the darker a dot in the main plot, the higher is
      the value of the radius, the $\phi$ parameter, or the $\alpha$
      parameter, respectively.}
     \label{main_results_2}
\end{figure*}

\section*{Acknowledgements}
Martin Lackner was supported by the Austrian Science Fund (FWF), project P31890.
Nimrod Talmon  was supported by the Israel Science Foundation (ISF; GrantNo.630/19). 
This project has received funding from the European Research Council
(ERC) under the European Union’s Horizon 2020 research and innovation
programme (grant agreement No 101002854).
\begin{center}
  \includegraphics[width=3cm]{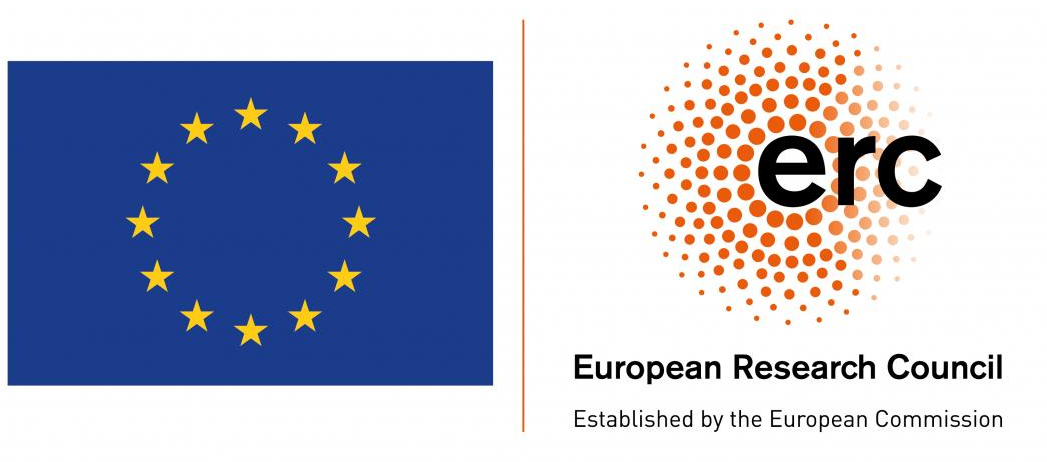}
\end{center}

\clearpage

% \bibliographystyle{named}
% \small % (if needed) we are allowed to use this for ijcai-22 
\bibliographystyle{plainnat}
\bibliography{bib.bib}

%%%%%%%%%%%%%%%%%%%%%%%%
%%%%%%%%%%%%%%%%%%%%%%%%
%%%%%%%%%%%%%%%%%%%%%%%%
%%%%%%%%%%%%%%%%%%%%%%%%
%%%%%%%%%%%%%%%%%%%%%%%%

\clearpage
%\newpage
%\quad

\appendix

\newcommand\finalwidthtmp{3.7cm}

\section{Proof of Proposition~\ref{hamming-hard}}

\begin{proof}[Proof of Proposition~\ref{hamming-hard}]
We describe a reduction from the NP-hard Clique problem.
An instance of Clique is $(G, k)$, where $G = (X, Y)$ is a graph, with set of vertices $X = \{x_1, \ldots, x_n\}$ and set of edges $Y$.

Given such an instance $(G, k)$ of Clique, create two elections, $A$ and $B$, with $n$ voters and $n$ candidates each.
The intuitive idea of the reduction is for $A$ to correspond to $G$ and for $B$ to correspond to a $k$-clique; then, a minimal Hamming distance would correspond to finding a $k$-clique in $A$. Details follows.

Denote the voters of $A$ by $\{Av_1, \ldots, Av_n\}$
and the candidates of $A$ by $\{Ac_1, \ldots, Ac_n\}$.
For each $i \in [n]$, let $Av_i$ approve $Ac_i$.
Furthermore, for each $i \in [n]$, let $Av_i$ approve each $Ac_j$ for which there is an edge $\{v_i, v_j\}$ in $G$.

Denote the voters of $B$ by $\{Bv_1, \ldots, Bv_n\}$ and the candidates of $B$ by $\{Bc_1, \ldots, Bc_n\}$.
For each $1 \leq k$, let $Bv_i$ approve $Bc_i$, and disapprove all other candidates.

Set the bound on the Hamming distance to be $D = |E| - k^2$.
This completes the description of the reduction.
Next we show the two directions for correctness.

First, assume that there is a $k$-clique in $G$. Without loss of generality let $x_{i_1}, \ldots, x_{i_k}$ be the vertices of the $k$-clique. Then, match the voters of $A$ so that $Av_j$, $j \in [k]$, is matched to $Bv_{i_j}$; and match the candidates of $A$ so that $Ac_j$, $j \in [k]$, is matched to $Bv_{i_j}$. That way, all the approvals of $B$ correspond to approvals in the resulting $A$; in particular, for each $Bv_i$ and $Bc_j$ such that $Bv_i$ approves $Bc_j$, the matched voter is approving the matched candidate. The number of approvals of the matched $A$ that correspond to disapprovals in $B$ are exactly $D$.

For the other direction, note that the Hamming distance corresponding to any matching of the voters and candidates of $A$ equals to $D' := D + Z$, where $Z$ is the number of disapprovals between the first $k$ matched voters and candidates, thus $Z = 0$ only when corresponding to a $k$-clique.
\end{proof}

\section{Cohesive Groups}

In this section we describe algorithms we used to obtain results regarding cohesive groups.
The algorithm is based on the one provided by Anonymous [2022].

\subsection{Calculating Maximum Cohesiveness Level}

We show an algorithm using Integer Linear Programming (ILP) that we used to calculate the maximum cohesiveness level of a group of voters. 

Let $E = (C, V)$ be an approval election instance and $k$ be the committee size. We ask what is the maximum value of $\ell$ such that there exists an $\ell$-cohesive group of voters, that is, a subset of at least $\nicefrac{\ell \cdot |V|}{k}$ voters that have at least $\ell$ common candidates, that is, approved by all of them.

First of all, let us point out that we can find $\ell$ by using binary search, because each $\ell$-cohesive group is also an $\ell'$-cohesive group for any $\ell' \leq \ell$. We can also search for $\ell$ by looping through each possible value, that is, from $1$ to $k$, and stop as soon as there is no cohesive group of a given the cohesiveness level. The second option may be better when we expect the cohesiveness level to be small, however, the first is better in the sense of theoretical computational complexity. In our algorithm, we focus on checking if there exists an $\ell$-cohesive group for a given cohesiveness level $\ell$.

Let us assume we are given an election $E = (C, V)$, a committee size $k$ and a cohesiveness level $\ell$. We show how to construct an ILP instance that indicates whether there exists an $\ell$-cohesive group in $E$. For the sake of brevity, we set $m = |C|$ and $n = |V|$. Furthermore, let $A$ be the binary matrix of approvals for~$E$, that is, we have $a_{ij} = 1$ if the $i$-th voter approves the $j$-th candidate, and we have $a_{ij} = 0$ otherwise.

We note that, if there exists an $\ell$-cohesive
group of any size, then there also exists an $\ell$-cohesive group of size exactly $s = \lceil \nicefrac{\ell \cdot n}{k} \rceil$, that is, with the lowest possible size (we can just remove redundant voters, because the set of commonly approved candidates would not decrease). Thus it is enough to ask whether there exists an $\ell$-cohesive group of size $s$.

To construct our ILP instance, we create the specified variables:
\begin{enumerate}
\item For each voter $v_i$, we create a binary variable $x_i$, with the intention that $x_i = 1$ if the $i$-th voter belongs to the cohesive group, and $x_i = 0$ otherwise. 
\item For each candidate $c_j$, we create a binary variable $y_i$, with the intention that $y_j = 1$ if all the selected voters (that is, specified by variables $x_1, \ldots, x_n$) approve the $j$-th candidate, and $y_j = 0$ otherwise.
\end{enumerate}
For convenience, we refer to the voters (to the candidates) whose $x_i$ ($y_j$) variables
are set to $1$ as \emph{selected}. 

Now let us specify the constraints for these variables. First of all, we need to ensure that exactly $s$ voters and at least $\ell$ candidates have been selected:
\begin{align*}
  & \textstyle \sum_{i=1}^{n} {x_i} = s, & \text{and}& &\textstyle \sum_{j=1}^{m} {y_j} \ge \ell.
\end{align*}
Furthermore, we need to ensure that each selected voter approves all the selected candidates. Thus, for each $j \in [m]$, we form a
constraint:
\[
  \textstyle \sum_{i=1}^{n} {a_{ij} \cdot x_i} \ge s \cdot y_j.
\]
Now let us show how the above inequality works. On the one hand, if the $j$-th candidate is not selected, then this inequality is
satisfied trivially, because the right-hand side is equal to $0$. On the other hand, if the $j$-th candidate is selected, then the sum on the left-hand side must be at least $s$, i.e., there must be at least $s$ selected voters who approve this candidate. Since there are exactly $s$ selected voters, all of them must approve the $j$-th candidate.

Now we see that, if there exists an assignment which satisfies the above constraints, then the selected voters form an $\ell$-cohesive group. Otherwise, there is no $\ell$-cohesive group.

\section{Correlation}\label{apdx:correlation}
The dataset used for comparing metrics in Figure~\ref{correlation} consists of: 40 elections from the Disjoint Models, 45 elections from the Noise Models with Hamming distance, 50 elections from the Truncated Urn Models, 50 elections from Euclidean Models, 134 elections from Resampling Models, 20 elections from IC,  20 elections from ID, and four extreme elections (i.e., 0.5-IC, 0.5-ID, Empty, Full).

\end{document}